\documentclass[11pt]{article}
\usepackage[margin=1in]{geometry}
\usepackage[colorlinks=true, allcolors=blue]{hyperref}

\usepackage{fixmath}
\usepackage{bm}
\usepackage{amsbsy}
\usepackage{color}
\usepackage{verbatim}
\usepackage{multirow}
\usepackage{amssymb}
\usepackage{amsthm}
\usepackage{amsmath}
\usepackage{array}
\usepackage{mathtools}

\usepackage{graphicx}
\usepackage{tikz}
\usetikzlibrary{positioning}
\usepackage{xcolor}

\usepackage{thm-restate}

\newtheorem{theorem}{Theorem}
\newtheorem{lemma}[theorem]{Lemma}

\newtheorem{remark}[theorem]{Remark}
\newtheorem{claim}[theorem]{Claim}

\newcommand{\size}[1]{\ensuremath{|#1|}}
\newcommand{\ceil}[1]{\ensuremath{\lceil#1\rceil}}
\newcommand{\floor}[1]{\ensuremath{\lfloor#1\rfloor}}

\newcommand{\lrA}[1]{\ensuremath{\left(#1\right)}}

\def\OPT{\mbox{OPT}}

\def\C{\mathcal{C}}

\def\P{\mathcal{P}}
\def\Q{\mathcal{Q}}
\def\B{\mathcal{B}}

\newcommand{\EE}[1]{\ensuremath{\mathbb{E}[#1]}}

\newcommand*\circled[1]{\tikz[baseline=(char.base)]{
            \node[shape=circle,draw,minimum size=4mm, inner sep=0pt] (char) {#1};}}

\title{An Improved Algorithm for a Bipartite Traveling Tournament in Interleague Sports Scheduling}

\author
{
Jingyang Zhao\\
University of Electronic Science and Technology of China\\
\texttt{jingyangzhao1020@gmail.com}
\and
Mingyu Xiao\\
University of Electronic Science and Technology of China\\
\texttt{myxiao@uestc.edu.cn}
}

\date{}

\begin{document}
\maketitle
\begin{abstract}
The bipartite traveling tournament problem (BTTP)
addresses inter-league sports scheduling, which aims to design a feasible bipartite tournament between two $n$-team leagues under some constraints such that the total traveling distance of all participating teams is minimized. Since its introduction, several methods have been developed to design feasible schedules for NBA, NPB and so on. In terms of solution quality with a theoretical guarantee, previously only a $(2+\varepsilon)$-approximation is known for the case that $n\equiv 0 \pmod 3$. Whether there are similar results for the cases that $n\equiv 1 \pmod 3$ and $n\equiv 2 \pmod 3$ was asked in the literature. In this paper, we answer this question positively by proposing a $(3/2+\varepsilon)$-approximation algorithm for any $n$ and any constant $\varepsilon>0$, which also improves the previous approximation ratio for the case that $n\equiv 0 \pmod 3$.
\end{abstract}

\maketitle

\section{Introduction}
The traveling tournament problem (TTP), introduced by Easton~\emph{et al.}~\cite{easton2001traveling}, is a well-known benchmark problem in the field of sports scheduling~\cite{kendall2010scheduling}.
This problem aims to find a double round-robin tournament, minimizing the total traveling distance of all participating teams.
In a double round-robin tournament involving $n$ teams (where $n$ is even), each team plays two games against each of the other $n-1$ teams that includes one home game at its own home venue and one away game at its opponent's home venue, and on each day each team can only play one game. Moreover, all games should be scheduled on $2(n-1)$ consecutive days, subject to several constraints on the maximum number of consecutive home/away games for each team, ensuring a balanced tournament arrangement. In the problem, we consider that the distance is a semi-metric, i.e., it satisfies the symmetry and triangle inequality properties.
An overview of TTP and its variants, along with their various applications on sports scheduling, 
can be found in~\cite{DBLP:journals/eor/BulckGSG20,duran2021sports}.

The bipartite traveling tournament problem (BTTP), introduced by Hoshino and Kawarabayashi~\cite{Bipartite-conference}, is an inter-league extension of TTP. Given two $n$-team leagues, it asks for a distance-optimal double round-robin bipartite tournament between the leagues, where every team in one league plays one home game and one away game against each team in another league.
It also requires that each team plays only one game on each day, and all games need to be scheduled on $2n$ consecutive days.
For TTP/BTTP, the double round-robin (bipartite) tournament is subject to the following three basic constraints or assumptions:
\begin{itemize}
\item \emph{No-repeat}: No pair of teams can play against each other in two consecutive games.
\item \emph{direct-traveling}: Each team travels directly from its game venue on the $i$-th day to its game venue on the $(i+1)$-th day, where we assume that all teams are initially at home before the first game starts and will return home after the last game ends.
\item \emph{Bounded-by-3}: Each team can play at most 3-consecutive home games or away games.
\end{itemize}
In the last constraint, we require that the maximum number of consecutive home/away games for each team is at most three. 
This is the most extensively studied 
case~\cite{lim2006simulated,anagnostopoulos2006simulated,goerigk2014solving}.
The case that at most two consecutive home/away games are allowed is also studied in some references~\cite{thielen2012approximation,DBLP:conf/atmos/ChatterjeeR21,DBLP:conf/ijcai/ZhaoX21}. 
A small number of consecutive home/away games requires teams to return home frequently, which may make the tournament balanced at a cost of a possible longer traveling distance.   

As indicated in a previous study~\cite{Bipartite-jornal}, BTTP has potential applications in sports scheduling, such as for events like the Davis Cup, the biennial Ryder Cup, the National Basketball Association (NBA), and the Nippon Professional Baseball (NPB). 
Specifically, the NPB inter-league scheduling problem is precisely BTTP for the case $n = 6$~\cite{Bipartite-jornal}.
Exploring BTTP also holds promise for advancing the theoretical development of enhanced algorithms for TTP. For instance, Zhao and Xiao~\cite{pathpacking2023} utilized an approximation algorithm for BTTP between two groups by grouping teams and achieved an efficient polynomial-time approximation
scheme for a special case of TTP where all teams are positioned along a line. BTTP possesses a simpler structure, and scheduling for TTP can be derived from scheduling for BTTP using a divide-and-conquer method~\cite{DBLP:conf/atmos/ChatterjeeR21,zmor}. Therefore, effective algorithms for BTTP have the potential to yield effective algorithms for TTP as well.

\subsection{Related Work}
Both TTP and BTTP are difficult optimization problems, and their NP-hardness has been established in~\cite{thielen2011complexity,Bipartite-conference}.
In the online benchmark~\cite{trick2007challenge,DBLP:journals/eor/BulckGSG20}, many instances of TTP with more than ten teams have not been completely solved even by using high-performance machines.
TTP and BTTP have been extensively studied both in theory~\cite{hoshino2012linear,hoshino2013approximation,westphal2014,DBLP:conf/mfcs/XiaoK16,imahori20211+,DBLP:conf/cocoon/ZhaoX21,zhao20255} and practice~\cite{easton2003solving,di2007composite,DBLP:conf/aaai/HentenryckV07,hoshino2011distance,Bipartite-conference,Bipartite-jornal,DBLP:journals/anor/GoerigkW16,frohner2023approaching}.

TTP and BTTP are also rich problems in approximation algorithms. An algorithm is called an $\alpha$-approximation algorithm if it can generate a feasible schedule in polynomial time such that the total traveling distance of all teams is within $\alpha$ times the optimal.
For TTP, Miyashiro \emph{et al.}~\cite{miyashiro2012approximation} first proposed a $(2+\varepsilon)$-approximation algorithm where $\varepsilon$ is an arbitrary fixed constant, and the ratio was later improved to $(1.667+\varepsilon)$ by Yamaguchi \emph{et al.}~\cite{yamaguchi2009improved}, and $(1.598+\varepsilon)$ by Zhao \emph{et al.}~\cite{zhao2022improved}.
For BTTP, the only known result is a $(2+\varepsilon)$-approximation algorithm for the case that $n\equiv0\pmod3$, proposed by Hoshino and Kawarabayashi~\cite{hoshino2013approximation}, and whether there exist similar algorithms for the cases that $n\equiv1\pmod3$ and $n\equiv2\pmod3$ was asked in this paper.


\subsection{Our Results}
In this paper, we design a new algorithm for BTTP and prove that our algorithm can achieve an approximation ratio of $(3/2+\varepsilon)$ for any $n$ and any constant $\varepsilon>0$. This not only positively answers the open question for the cases that $n\equiv1\pmod3$ and $n\equiv2\pmod3$  but also improves the previous best ratio of $(2+\varepsilon)$ for the case that 
$n\equiv0\pmod3$~\cite{hoshino2013approximation}.

To achieve our improvement, in Section~\ref{SC.3}, we propose a novel \emph{3-path construction}: given any (feasible) 3-path packing, it can compute a feasible schedule for BTTP in polynomial time.
In Section~\ref{SC.4}, we do some analysis. 
Specifically, in Section~\ref{SC.4.2}, we introduce a novel lower bound for any optimal solution to BTTP, which is related to the minimum weight cycle packing (i.e., a set of cycles where each cycle has an order of at least 3). 
In Section~\ref{SC.4.1}, we analyze the quality of the 3-path construction and show that the quality of the solution produced by the 3-path construction is closely related to the quality of the used 3-path packing. Therefore, the 3-path construction reduces BTTP to the task of finding a good 3-path packing.
In Section~\ref{SC.4.3}, we use the minimum weight cycle packing to obtain a good 3-path packing and show that by applying this 3-path packing to our 3-path construction, we can get a schedule with an approximation ratio of at most $(3/2+\varepsilon)$ in polynomial time.

In addition to the theoretical results, in Section~\ref{SC.5}, we explore the practical applications of our 3-path construction. 
We create a new instance based on the real situation of NBA for BTTP with 32 teams ($n=16$) and test our algorithm on this instance. Note that the previous algorithm~\cite{hoshino2013approximation} only works for the case that $n\equiv0\pmod3$, while it holds $n\equiv1\pmod3$ in the new instance. 
Then, we extend our 3-path construction to a more efficient \emph{3-cycle construction}. 
By using simple swapping heuristics and applying these two constructions to the new instance, experimental results show that our algorithms can be implemented within 2 seconds, and the quality of the solutions is much better than the expected $1.5$-approximation ratio. Specifically, by using the well-known \emph{Independent Lower Bound}~\cite{easton2003solving}, we show that the gaps between our results and the optimal are at most $24.66\%$ and $9.42\%$ for the 3-path and 3-cycle constructions, respectively.

\section{Notations}
An instance of BTTP can be presented by a complete graph $G=(V=X\cup Y, E, w)$ with $2n$ vertices representing $2n$ teams, where $X=\{x_0,\dots,x_{n-1}\}$ and $Y=\{y_0,\dots,y_{n-1}\}$ are two $n$-team leagues, and $w$ is a non-negative semi-metric weight function on the edges in $E$ that satisfies the symmetry and triangle inequality properties, i.e., $w(x,y)+w(y,z)\geq w(x,z)=w(z,x)$ for any $x,y,z\in V$. 
The weight $w(x,y)$ of edge $xy\in E$ represents the distance between the homes of teams $x$ and $y$. 
We also extend the function to a set of edges, i.e., we let $w(E')\coloneqq\sum_{e\in E'} w(e)$ for any $E'\subseteq E$, and to a subgraph $G'$ of $G$, i.e., we let $w(G')$ be the total weight of all edges in $G'$. 
Given any $V'\subseteq V$, the complete graph induced by $V'$ is denoted by $G[V']$.

Given two vertices/teams $x\in X$ and $y\in Y$, and two sets $X'\subseteq X$ and $Y'\subseteq Y$, we define the following notations. 
We use $x\rightarrow y$ to denote a game between $x$ and $y$ at the home of $y$, and $x\leftrightarrow y$ to denote two games between $x$ and $y$ including one game at the home of $x$ and one game at the home of $y$.
We use $E_{Y'}(x)$ (resp., $E_{X'}(y)$) to denote the set of edges in $G$ between the vertex $x$ (resp., $y$) and one vertex in $Y'$ (resp., $X'$), i.e., $E_{Y'}(x)=\{xy\mid y\in Y'\}$.
We also let $\delta_{Y'}(x) = w(E_{Y'}(x))$ and $\delta_{Y'}(X')=\sum_{x\in X'}\delta_{Y'}(x)$ (resp., $\delta_{X'}(y) = w(E_{X'}(y))$ and $\delta_{X'}(Y')=\sum_{y\in Y'}\delta_{X'}(y)$). Note that $\delta_{Y'}(X')=\delta_{X'}(Y')$.

Two subgraphs or sets of edges are \emph{vertex-disjoint} if they do not share a common vertex.
An $l$-cycle $C=x_1x_2\dots x_lx_1$ is a simple cycle on $l$ different vertices $\{x_1,\dots,x_l\}$. 
It consists of $l$ edges $\{x_1x_2,\dots,x_lx_1\}$, and its \emph{order} is defined as $\size{C}\coloneqq l$, representing the number of vertices it contained.
A cycle packing in a graph is a set of vertex-disjoint cycles, where the order of each cycle is at least three and the cycles cover all vertices of the graph.
The minimum weight cycle packing in a graph with $n$ vertices can be found in $O(n^3)$ time~\cite{hartvigsen1984extensions}.
Similarly, an $l$-path $x_1x_2\dots x_l$ is a simple path on $l$ different vertices $\{x_1,\dots,x_l\}$. It consists of $l-1$ edges $\{x_1x_2,\dots,x_{l-1}x_l\}$, and the vertices $x_1$ and $x_l$ are called its \emph{terminals}.
An \emph{$l$-path packing} in a graph is a set of vertex-disjoint $l$-paths, where the paths cover all vertices of the graph.

A \emph{walk} is a sequence of vertices where each consecutive pair of vertices is connected by an edge, and its weight is defined as the total weight of the edges traversed in the sequence.
A walk is \emph{closed} if the first and the last vertices are the same.
In a solution of BTTP, every team $v\in X\cup Y$ has an \emph{itinerary}, which is a closed walk starting and ending at $v$. 
This itinerary could be decomposed into several minimal closed walks, each starting and ending at $v$, referred to as \emph{trips}. Each trip is an $l$-cycle containing $v$ with $2\leq l\leq 4$, and all trips/cycles share only one common vertex $v$.
For example, Table~\ref{example-solution} shows a solution for leagues $X=\{x_0,x_1,x_2\}$ and $Y=\{y_0,y_1,y_2\}$, where home games are marked in bold. We can get that $x_2$ has an itinerary $x_2y_0y_1y_2x_2$, and $y_2$ has an itinerary $y_2x_1x_2y_2x_0y_2$, and $y_2$ has two trips $y_2x_1x_2y_2$ and $y_2x_0y_2$ (a 3-cycle and a 2-cycle sharing $y_2$ only).

\begin{table}[ht]
\centering
\begin{tabular}{c|*{6}{c}}
  & $0$ & $1$ & $2$ & $3$ & $4$ & $5$\\
\hline
  $x_0$ & $y_0$ & $y_1$ & $y_2$ & $\pmb{y_0}$ & $\pmb{y_1}$ & $\pmb{y_2}$\\
  $x_1$ & $\pmb{y_2}$ & $y_0$ & $y_1$ & $y_2$ & $\pmb{y_0}$ & $\pmb{y_1}$\\
  $x_2$ & $\pmb{y_1}$ & $\pmb{y_2}$ & $y_0$ & $y_1$ & $y_2$ & $\pmb{y_0}$\\
\hline
  $y_0$ & $\pmb{x_0}$ & $\pmb{x_1}$ & $\pmb{x_2}$ & $x_0$ & $x_1$ & $x_2$\\
  $y_1$ & $x_2$ & $\pmb{x_0}$ & $\pmb{x_1}$ & $\pmb{x_2}$ & $x_0$ & $x_1$\\
  $y_2$ & $x_1$ & $x_2$ & $\pmb{x_0}$ & $\pmb{x_1}$ & $\pmb{x_2}$ & $x_0$\\
\end{tabular}
\caption{
A solution for BTTP with two leagues $X=\{x_0,x_1,x_2\}$ and $Y=\{y_0,y_1,y_2\}$, where home games are marked in bold.
}
\label{example-solution}
\end{table}

Fix a constant $\varepsilon>0$.
We want to consider a 3-path packing in $G$. Since $n$ may not be divisible by 3, we will remove a small number of vertices (as we will see that the number is related to $\varepsilon$) in $X$ and $Y$ to obtain two new sets $X_\varepsilon\subseteq X$ and $Y_\varepsilon\subseteq Y$ such that the number of vertices in $X_\varepsilon$ and $Y_\varepsilon$, denoted as $n_\varepsilon$, are the same and are divisible by 3. Hence, there always exist 3-path packings in $G[X_\varepsilon]$ and $G[Y_\varepsilon]$. We also let $\overline{X_\varepsilon}\coloneqq X\setminus X_\varepsilon$ and $\overline{Y_\varepsilon}\coloneqq Y\setminus Y_\varepsilon$. The graph $G_\varepsilon\coloneqq G[X_\varepsilon\cup Y_\varepsilon]$ is called the \emph{core} graph of $G$ because we will see that the quality of our schedule is only dominated by the total traveling distance of teams in $X_\varepsilon\cup Y_\varepsilon$, i.e., the total traveling distance of teams in $\overline{X_\varepsilon}\cup\overline{Y_\varepsilon}$ is small. We denote the minimum weight cycle packing in $G[X_\varepsilon]$ (resp., $G[Y_\varepsilon]$) by $\C_{X_\varepsilon}$ (resp., $\C_{Y_\varepsilon}$).
And, we denote the weight (i.e., the total traveling distance of all teams) of an optimal solution for BTTP by $\OPT$. 


\section{The 3-Path Construction of the Schedule}\label{SC.3}
Assuming we are given a 3-path packing $\P_{X_\varepsilon}$ in $G[X_\varepsilon]$ and a 3-path packing $\P_{Y_\varepsilon}$ in $G[Y_\varepsilon]$, to construct a schedule that minimizes the total traveling distance of teams in $X_\varepsilon \cup Y_\varepsilon$, the main idea of the 3-path construction is to ensure every team in $X_\varepsilon$ (res., $Y_\varepsilon$) plays 3-consecutive away games along every 3-path in $\P_{Y_\varepsilon}$ (resp., $\P_{X_\varepsilon}$) from one terminal to another.

Consider an \emph{ideal} schedule where every team in $Y_\varepsilon$ (resp., $X_\varepsilon$) plays 3-consecutive away games along every 3-path in $\P_{X_\varepsilon}$ (resp., $\P_{Y_\varepsilon}$). Then, every team in $Y_\varepsilon$ (resp., $X_\varepsilon$) plays $\frac{n_\varepsilon}{3}$ away-trips with teams in $X_\varepsilon$ (resp., $Y_\varepsilon$). The total traveling distance of teams in $Y_\varepsilon$ for these away-trips is
\begin{align*}
&\sum_{y\in Y_\varepsilon}\sum_{{P=xx'x'' \in\P_{X_\varepsilon}}}(w(y,x)+w(y,x'')+w(P))=\sum_{xx'x''\in\P_{X_\varepsilon}}(\delta_{Y_\varepsilon}(x)+\delta_{Y_\varepsilon}(x''))+n_\varepsilon w(\P_{X_\varepsilon}).   
\end{align*}
Denote $\sum_{xx'x''\in\P_{X_\varepsilon}}(\delta_{Y_\varepsilon}(x)+\delta_{Y_\varepsilon}(x''))$ by $\delta_{Y_\varepsilon}(\P_{X_\varepsilon})$ for the sake of presentation, and define $\delta_{X_\varepsilon}(\P_{Y_\varepsilon})$ in the similar manner. 
The total traveling distance of teams in $X_\varepsilon\cup Y_\varepsilon$ for these away-trips is $\delta_{Y_\varepsilon}(\P_{X_\varepsilon})+n_\varepsilon w(\P_{X_\varepsilon})+\delta_{X_\varepsilon}(\P_{Y_\varepsilon})+n_\varepsilon w(\P_{Y_\varepsilon})$.

Our 3-path construction will generate a schedule such that the total traveling distance of teams in $X\cup Y$ is close to 
$\delta_{Y_\varepsilon}(\P_{X_\varepsilon})+n_\varepsilon w(\P_{X_\varepsilon})+\delta_{X_\varepsilon}(\P_{Y_\varepsilon})+n_\varepsilon w(\P_{Y_\varepsilon})$, i.e, the performance of our schedule is close to the ideal schedule. To achieve this, we make sure that almost every team in $Y_\varepsilon$ (resp., $X_\varepsilon$) plays 3-consecutive away games along every 3-path in $\P_{X_\varepsilon}$ (resp., $\P_{Y_\varepsilon}$).
Next, we give our 3-path construction in details.

Let $\varepsilon>0$ be a fixed constant. 
Then, let $d\coloneqq 6\lceil 1/\varepsilon \rceil$, and $m\coloneqq2\floor{\frac{n}{6d}}-1$ (it is useful to ensure that $d$ is even and $m$ is odd). We assume that $n\geq 108d^3$, as otherwise, we have $n< 108d^3=O_\varepsilon(1)$\footnote{$O_\varepsilon(1)$ means a constant related to $\varepsilon$.}, and in this case, we can solve the problem optimally in constant time.
Let $n_\varepsilon\coloneqq 3md$ and $l\coloneqq n-n_\varepsilon$. Since $n\geq 108d^3$, we can get $m\geq 18d^2$.
Moreover, since $l=n-n_\varepsilon=n-3md=n-6d\floor{\frac{n}{6d}}+3d$, we can get 
\begin{equation}\label{eq0}
3d\leq l\leq 9d\leq 18d^2\leq m.
\end{equation}

We will select $n_\varepsilon$ vertices from $X$ (resp., $Y$) to form $X_\varepsilon$ (resp., $Y_\varepsilon$). The details of how to select these vertices is deferred to the analysis part of our schedule.
Recall that $\P_{X_\varepsilon}$ and $\P_{Y_\varepsilon}$ are two given 3-path packings in $G[X_\varepsilon]$ and $G[Y_\varepsilon]$. Since $n_\varepsilon=3md$, both $\P_{X_\varepsilon}$ and $\P_{Y_\varepsilon}$ contain $md$ 3-paths. 

Assume that $\P_{X_\varepsilon}=\{P_1,\dots,P_{md}\}$ and $\P_{Y_\varepsilon}=\{Q_1,\dots,Q_{md}\}$, where $P_i=x_{3i-3}x_{3i-2}x_{3i-1}$ and $Q_i=y_{3i-3}y_{3i-2}y_{3i-1}$. Each 3-path corresponds to three teams. For every $d$ 3-paths from $P_1$ (resp., $Q_1$) to $P_{md}$ (resp., $Q_{md}$), we pack the corresponding $3d$ teams as a \emph{super-team}, denoted by $S_i$ (resp., $T_i$). So, $S_i=\{x_{3i'-3},x_{3i'-2},x_{3i'-1}\}_{i'=id-d+1}^{id}$. Recall that $l= n-n_\varepsilon$.
There are $l$ teams in $\overline{X_\varepsilon}$ (resp., $\overline{Y_\varepsilon}$), and we label them using $\{x_{n_\varepsilon},\dots, x_{n-1}\}$ (resp., $\{y_{n_\varepsilon},\dots, y_{n-1}\}$). 
Thus, we can get $l$ \emph{team-pairs} $\{L_1,\dots,L_l\}$ where $L_i=\{x_{n_\varepsilon+i-1},y_{n_\varepsilon+i-1}\}$. 


The main idea of the 3-path construction is that: we first arrange a schedule of \emph{super-games} between super-teams (including the team-pairs); then, we extend the super-games into regular games between regular teams, which will form a feasible schedule for BTTP.

Our schedule contains $m+1$ time slots, where we have $m$ super-games in each of the first $m$ time slots, including $l$ \emph{left super-games} and $m-l$ \emph{normal super-games}. Each super-team will attend one super-game in each of the first $m$ time slots. The last time slot is special, and we will explain it later. 
Each of the first $m$ time slots spans $6d$ days and the last time slot spans $2l$ days. Hence, our schedule will span $6md+2l=2n_\varepsilon+2l=2n$ days.

\subsection{The First $m$ Time Slots}
In the first time slot, the super-games are arranged as shown in Figure~\ref{fig01}, where $m=5$ and $l=2$. 
Each super-team is denoted by a cycle node and each team-pair is denoted by a square node. 
There are $2m$ super-teams, $l$ team-pairs, and $m$ super-games.
Each super-game is denoted by an edge between two super-teams. 
The most left $l$ super-games, each involving two super-teams and one team-pair, are called the \emph{left super-games}. The left super-game involving team-pair $L_i$ is denoted as the \emph{$i$-th left super-game}.
The other $m-l$ super-games are called \emph{normal super-games}.
It is worth noting that there are only a constant number of left super-games since $l\leq 9d=O_\varepsilon(1)$. So, in our 3-path construction almost all super-games are normal super-games.

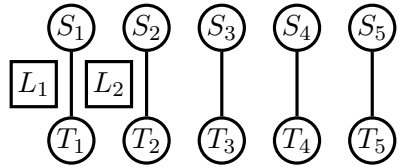
\begin{figure}[ht]
\centering
\begin{tikzpicture}
\tikzstyle{leftsuperteam}=[circle, draw=black!100, very thick, minimum size=6mm, inner sep=0pt]
\tikzstyle{normalsuperteam}=[circle, draw=black!100, very thick, minimum size=6mm, inner sep=0pt]
\tikzstyle{rightsuperteam}=[draw=black!100, very thick, minimum size=6mm, inner sep=0pt]

\node[normalsuperteam]      (up)    at (1, 0.75) {$S_1$};
\node[normalsuperteam]      (down)    at (1, -0.75) {$T_1$};
\node[rightsuperteam]      (middle)    at (0.5, 0)  {$L_1$};
\draw[very thick,-] (up.south) to (down.north);

\node[normalsuperteam]      (up)    at (2, 0.75) {$S_2$};
\node[normalsuperteam]      (down)    at (2, -0.75) {$T_2$};
\node[rightsuperteam]      (middle)    at (1.5, 0)  {$L_2$};
\draw[very thick,-] (up.south) to (down.north);

\node[normalsuperteam]      (up)    at (3, 0.75) {$S_3$};
\node[normalsuperteam]      (down)    at (3, -0.75) {$T_3$};
\draw[very thick,-] (up.south) to (down.north);

\node[normalsuperteam]      (up)    at (4, 0.75) {$S_4$};
\node[normalsuperteam]      (down)    at (4, -0.75) {$T_4$};
\draw[very thick,-] (up.south) to (down.north);

\node[normalsuperteam]      (up)    at (5, 0.75) {$S_5$};
\node[normalsuperteam]      (down)    at (5, -0.75) {$T_5$};
\draw[very thick,-] (up.south) to (down.north);
\end{tikzpicture}
\caption{The super-game schedule in the first time slot, where $m=5$ and $l=2$.}
\label{fig01}
\end{figure}

In the second time slot, the super-games are arranged as shown in Figure~\ref{fig02}.
Compared to the first time slot, the positions take one leftward shift for super-teams $S_1,\dots, S_m$, take two leftward shifts for super-teams $T_1,\dots, T_m$, and keep unchanged for team-pairs $L_1,\dots, L_l$.

\begin{figure}[ht]
\centering
\begin{tikzpicture}
\tikzstyle{leftsuperteam}=[circle, draw=black!100, very thick, minimum size=6mm, inner sep=0pt]
\tikzstyle{normalsuperteam}=[circle, draw=black!100, very thick, minimum size=6mm, inner sep=0pt]
\tikzstyle{rightsuperteam}=[draw=black!100, very thick, minimum size=6mm, inner sep=0pt]

\node[normalsuperteam]      (up)    at (1, 0.75) {$S_2$};
\node[normalsuperteam]      (down)    at (1, -0.75) {$T_3$};
\node[rightsuperteam]      (middle)    at (0.5, 0)  {$L_1$};
\draw[very thick,-] (up.south) to (down.north);

\node[normalsuperteam]      (up)    at (2, 0.75) {$S_3$};
\node[normalsuperteam]      (down)    at (2, -0.75) {$T_4$};
\node[rightsuperteam]      (middle)    at (1.5, 0)  {$L_2$};
\draw[very thick,-] (up.south) to (down.north);

\node[normalsuperteam]      (up)    at (3, 0.75) {$S_4$};
\node[normalsuperteam]      (down)    at (3, -0.75) {$T_5$};
\draw[very thick,-] (up.south) to (down.north);

\node[normalsuperteam]      (up)    at (4, 0.75) {$S_5$};
\node[normalsuperteam]      (down)    at (4, -0.75) {$T_1$};
\draw[very thick,-] (up.south) to (down.north);

\node[normalsuperteam]      (up)    at (5, 0.75) {$S_1$};
\node[normalsuperteam]      (down)    at (5, -0.75) {$T_2$};
\draw[very thick,-] (up.south) to (down.north);
\end{tikzpicture}
\caption{The super-game schedule in the second slot, where $m=5$ and $l=2$.}
\label{fig02}
\end{figure}
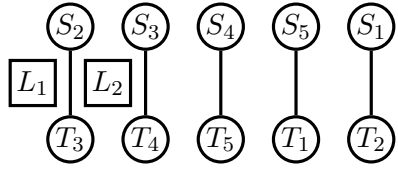

The schedules for the first $m$ time slots are derived analogously. We have two observations.

\begin{claim}\label{onesuper}
During the first $m$ time slots, there is only one super-game between super-teams $S_i$ and $T_j$ for every $1 \leq i, j \leq m$.
\end{claim}
\begin{proof}
By the 3-path construction, there is a super-game between super-teams $S_i$ and $T_j$ if and only if the schedule is at the time slot of $(j-i+1+m)\bmod m$.
\end{proof}

\begin{claim}\label{unique}
During the first $m$ time slots, each super-team participates in only one $i$-th left super-game for every $1 \leq i\leq l$.
\end{claim}
\begin{proof}
From the 1-st to the $m$-th time slot, the super-team in $\{S_1,\dots, S_m\}$ playing an $i$-th left super-game is $S_i,S_{i+1},\dots,S_m,S_1,\dots,S_{i-1}$, respectively; the super-team in $\{T_1,\dots, T_m\}$ playing an $i$-th left super-game is $T_i,T_{i+2},\dots,T_{i-1},T_{i+1},\dots,T_{i-2}$, respectively, since $m$ is odd. So, each super-team plays only one $i$-th left super-game.
\end{proof}

Next, we show how to extend normal super-games and left super-games into regular games. 

We consider a super-game between super-teams $S_i$ and $T_j$. For the sake of presentation, we define $s_{k}\coloneqq x_{(i-1)d+k}$ and $t_{k}\coloneqq y_{(j-1)d+k}$ for any $0\leq k<3d$.
By definitions of $S_i$ and $T_j$, we have that $S_i=\{s_0,s_1,s_2,\dots, s_{3d-3},s_{3d-2},s_{3d-1}\}$ and $T_j=\{t_0,t_1,t_2,\dots,t_{3d-3},t_{3d-2},t_{3d-1}\}$.

First consider that the super-game is a normal super-game.

\textbf{Normal super-games:}
The normal super-game will be extended into regular games on $6d$ days in the following way:
\begin{itemize}
    \item Team $s_{3i+i'}$ plays an away game with $t_{3j+j'}$ on day $(6(i+j)+i'+j')\bmod 6d$,
    \item Team $s_{3i+i'}$ plays a home game with $t_{3j+j'}$ on day $(6(i+j)+i'+j'+3)\bmod 6d$,
\end{itemize}
where $0\leq i,j\leq d-1$ and $0\leq i',j'\leq 2$. An illustration of the regular games after extending one normal super-game for $d=2$ is shown in Table~\ref{normal-example}.

\begin{table}[ht]
\centering
\begin{tabular}{c|*{12}{c}}
  & $0$ & $1$ & $2$ & $3$ & $4$ & $5$ & $6$ & $7$ & $8$ & $9$ & $10$ & $11$\\
\hline
  $s_0$ & $t_0$ & $t_1$ & $t_2$ & $\pmb{t_0}$ & $\pmb{t_1}$ & $\pmb{t_2}$ & $t_3$ & $t_4$ & $t_5$ & $\pmb{t_3}$ & $\pmb{t_4}$ & $\pmb{t_5}$\\
  $s_1$ & $\pmb{t_5}$ & $t_0$ & $t_1$ & $t_2$ & $\pmb{t_0}$ & $\pmb{t_1}$ & $\pmb{t_2}$ & $t_3$ & $t_4$ & $t_5$ & $\pmb{t_3}$ & $\pmb{t_4}$\\
  $s_2$ & $\pmb{t_4}$ & $\pmb{t_5}$ & $t_0$ & $t_1$ & $t_2$ & $\pmb{t_0}$ & $\pmb{t_1}$ & $\pmb{t_2}$ & $t_3$ & $t_4$ & $t_5$ & $\pmb{t_3}$\\
  $s_3$ & $t_3$ & $t_4$ & $t_5$ & $\pmb{t_3}$ & $\pmb{t_4}$ & $\pmb{t_5}$ & $t_0$ & $t_1$ & $t_2$ & $\pmb{t_0}$ & $\pmb{t_1}$ & $\pmb{t_2}$\\
  $s_4$ & $\pmb{t_2}$ & $t_3$ & $t_4$ & $t_5$ & $\pmb{t_3}$ & $\pmb{t_4}$ & $\pmb{t_5}$ & $t_0$ & $t_1$ & $t_2$ & $\pmb{t_0}$ & $\pmb{t_1}$\\
  $s_5$ & $\pmb{t_1}$ & $\pmb{t_2}$ & $t_3$ & $t_4$ & $t_5$ & $\pmb{t_3}$ & $\pmb{t_4}$ & $\pmb{t_5}$ & $t_0$ & $t_1$ & $t_2$ & $\pmb{t_0}$\\
\hline
  $t_0$ & $\pmb{s_0}$ & $\pmb{s_1}$ & $\pmb{s_2}$ & $s_0$ & $s_1$ & $s_2$ & $\pmb{s_3}$ & $\pmb{s_4}$ & $\pmb{s_5}$ & $s_3$ & $s_4$ & $s_5$\\
  $t_1$ & $s_5$ & $\pmb{s_0}$ & $\pmb{s_1}$ & $\pmb{s_2}$ & $s_0$ & $s_1$ & $s_2$ & $\pmb{s_3}$ & $\pmb{s_4}$ & $\pmb{s_5}$ & $s_3$ & $s_4$\\
  $t_2$ & $s_4$ & $s_5$ & $\pmb{s_0}$ & $\pmb{s_1}$ & $\pmb{s_2}$ & $s_0$ & $s_1$ & $s_2$ & $\pmb{s_3}$ & $\pmb{s_4}$ & $\pmb{s_5}$ & $s_3$\\
  $t_3$ & $\pmb{s_3}$ & $\pmb{s_4}$ & $\pmb{s_5}$ & $s_3$ & $s_4$ & $s_5$ & $\pmb{s_0}$ & $\pmb{s_1}$ & $\pmb{s_2}$ & $s_0$ & $s_1$ & $s_2$\\
  $t_4$ & $s_2$ & $\pmb{s_3}$ & $\pmb{s_4}$ & $\pmb{s_5}$ & $s_3$ & $s_4$ & $s_5$ & $\pmb{s_0}$ & $\pmb{s_1}$ & $\pmb{s_2}$ & $s_0$ & $s_1$\\
  $t_5$ & $s_1$ & $s_2$ & $\pmb{s_3}$ & $\pmb{s_4}$ & $\pmb{s_5}$ & $s_3$ & $s_4$ & $s_5$ & $\pmb{s_0}$ & $\pmb{s_1}$ & $\pmb{s_2}$ & $s_0$\\
\end{tabular}
\caption{
Extending the normal super-game between $\{s_0,s_1,s_2,s_3,s_4,s_5\}$ and $\{t_0,t_1,t_2$, $t_3,t_4,t_5\}$ into regular games on $6d=12$ days, where $d=2$ and home games are marked in bold.
}
\label{normal-example}
\end{table}

The design of normal super-games is essentially the same as the construction in~\cite{hoshino2013approximation}, which only works for the case that $n\equiv0\pmod 3$. It guarantees that all games between one team in $S_i$ and one team in $T_j$ are arranged. To get a feasible schedule for
the cases that $n\equiv1\pmod 3$ and $n\equiv2\pmod 3$, we need to design left super-games, which are newly introduced. 

\textbf{Left super-games:}
In this case, we have a team-pair $L_{i'}$, so we also define $s_{3d}\coloneqq x_{n_\varepsilon+i'-1}$ and $t_{3d}\coloneqq y_{n_\varepsilon+i'-1}$.
By definition of $L_{i'}$, we can get $L_{i'}=\{s_{3d},t_{3d}\}$. 
For all teams in $S_i\cup T_j\cup L_{i'}$, we define a set of games as 
\[
m_i\coloneqq\{s_{i'}\rightarrow t_{(i'+i)\bmod (3d+1)}\}_{i'=0}^{3d},
\] 
where each team in $S_i\cup T_j\cup L_{i'}$ plays one game. Define $\overline{m_i}$ as the set of games in $m_i$ but with reversed venues, and let $\overline{\overline{m_i}}\coloneqq m_i$.
Since $d$ is even, the extended regular games on $6d$ days can be presented by
\[
m_1\overline{m_2}m_3\overline{m_4}\cdots m_{3d-1}\overline{m_{3d}} \cdot\overline{\overline{m_2}m_3\overline{m_4}m_5\cdots \overline{m_{3d}}m_1}.
\]
An illustration of the regular games after extending one left super-game for $d=2$ is shown in Table~\ref{left-example}.

Note that if $d$ is odd, the games can be presented by
\[
m_1\overline{m_2}m_3\overline{m_4}\cdots m_{3d} \cdot\overline{m_1\overline{m_2}m_3\overline{m_4}\cdots m_{3d}}.
\]

\begin{table}[ht]
\centering
\begin{tabular}{c|*{12}{c}}
  & $0$ & $1$ & $2$ & $3$ & $4$ & $5$ & $6$ & $7$ & $8$ & $9$ & $10$ & $11$\\
\hline
  $s_0$ & $t_1$ & $\pmb{t_2}$ & $t_3$ & $\pmb{t_4}$ & $t_5$ & $\pmb{t_6}$ & $t_2$ & $\pmb{t_3}$ & $t_4$ & $\pmb{t_5}$ & $t_6$ & $\pmb{t_1}$\\
  $s_1$ & $t_2$ & $\pmb{t_3}$ & $t_4$ & $\pmb{t_5}$ & $t_6$ & $\pmb{t_0}$ & $t_3$ & $\pmb{t_4}$ & $t_5$ & $\pmb{t_6}$ & $t_0$ & $\pmb{t_2}$\\
  $s_2$ & $t_3$ & $\pmb{t_4}$ & $t_5$ & $\pmb{t_6}$ & $t_0$ & $\pmb{t_1}$ & $t_4$ & $\pmb{t_5}$ & $t_6$ & $\pmb{t_0}$ & $t_1$ & $\pmb{t_3}$\\
  $s_3$ & $t_4$ & $\pmb{t_5}$ & $t_6$ & $\pmb{t_0}$ & $t_1$ & $\pmb{t_2}$ & $t_5$ & $\pmb{t_6}$ & $t_0$ & $\pmb{t_1}$ & $t_2$ & $\pmb{t_4}$\\
  $s_4$ & $t_5$ & $\pmb{t_6}$ & $t_0$ & $\pmb{t_1}$ & $t_2$ & $\pmb{t_3}$ & $t_6$ & $\pmb{t_0}$ & $t_1$ & $\pmb{t_2}$ & $t_3$ & $\pmb{t_5}$\\
  $s_5$ & $t_6$ & $\pmb{t_0}$ & $t_1$ & $\pmb{t_2}$ & $t_3$ & $\pmb{t_4}$ & $t_0$ & $\pmb{t_1}$ & $t_2$ & $\pmb{t_3}$ & $t_4$ & $\pmb{t_6}$\\
  $s_6$ & $t_0$ & $\pmb{t_1}$ & $t_2$ & $\pmb{t_3}$ & $t_4$ & $\pmb{t_5}$ & $t_1$ & $\pmb{t_2}$ & $t_3$ & $\pmb{t_4}$ & $t_5$ & $\pmb{t_0}$\\
\hline
  $t_0$ & $\pmb{s_6}$ & $s_5$ & $\pmb{s_4}$ & $s_3$ & $\pmb{s_2}$ & $s_1$ & $\pmb{s_5}$ & $s_4$ & $\pmb{s_3}$ & $s_2$ & $\pmb{s_1}$ & $s_6$\\
  $t_1$ & $\pmb{s_0}$ & $s_6$ & $\pmb{s_5}$ & $s_4$ & $\pmb{s_3}$ & $s_2$ & $\pmb{s_6}$ & $s_5$ & $\pmb{s_4}$ & $s_3$ & $\pmb{s_2}$ & $s_0$\\
  $t_2$ & $\pmb{s_1}$ & $s_0$ & $\pmb{s_6}$ & $s_5$ & $\pmb{s_4}$ & $s_3$ & $\pmb{s_0}$ & $s_6$ & $\pmb{s_5}$ & $s_4$ & $\pmb{s_3}$ & $s_1$\\
  $t_3$ & $\pmb{s_2}$ & $s_1$ & $\pmb{s_0}$ & $s_6$ & $\pmb{s_5}$ & $s_4$ & $\pmb{s_1}$ & $s_0$ & $\pmb{s_6}$ & $s_5$ & $\pmb{s_4}$ & $s_2$\\
  $t_4$ & $\pmb{s_3}$ & $s_2$ & $\pmb{s_1}$ & $s_0$ & $\pmb{s_6}$ & $s_5$ & $\pmb{s_2}$ & $s_1$ & $\pmb{s_0}$ & $s_6$ & $\pmb{s_5}$ & $s_3$\\
  $t_5$ & $\pmb{s_4}$ & $s_3$ & $\pmb{s_2}$ & $s_1$ & $\pmb{s_0}$ & $s_6$ & $\pmb{s_3}$ & $s_2$ & $\pmb{s_1}$ & $s_0$ & $\pmb{s_6}$ & $s_4$\\
  $t_6$ & $\pmb{s_5}$ & $s_4$ & $\pmb{s_3}$ & $s_2$ & $\pmb{s_1}$ & $s_0$ & $\pmb{s_4}$ & $s_3$ & $\pmb{s_2}$ & $s_1$ & $\pmb{s_0}$ & $s_5$\\
\end{tabular}
\caption
{
Extending the left super-game between $\{s_0,s_1,s_2,s_3,s_4,s_5\}$ and $\{t_0,t_1,t_2$, $t_3,t_4,t_5\}$ (including a team-pair $\{s_6,t_6\}$) into regular games on $6d=12$ days, where $d=2$ and home games are marked in bold.
}
\label{left-example}
\end{table}

For each left super-game, we have two more teams due to a team-pair, and then two days' games in $m_0$ and $\overline{m_0}$ are still unarranged in extending the super-game on $6d$ days.

\subsection{The Last Time Slot}
In the last time slot, we will arrange all unarranged games due to left super-games. First, we will consider the unarranged games involving teams in team-pairs. Then, we will consider the unarranged games involving teams in super-teams.

\textbf{The teams in team-pairs.} Recall that the $l$ team-pairs contain $l$ teams in $\overline{X_\varepsilon}=\{x_{n_\varepsilon},\dots,x_{n-1}\}$ and $l$ teams in $\overline{Y_\varepsilon}=\{y_{n_\varepsilon},\dots,y_{n-1}\}$. The next claim shows the unarranged regular games involving teams in team-pairs. 
\begin{claim}\label{unarranged-1}
The unarranged regular games involving teams in team-pairs are $\{x_i\leftrightarrow y_j\mid x_i\in \overline{X_\varepsilon}, y_j\in\overline{Y_\varepsilon}\}$.
\end{claim}
\begin{proof}
Since in our 3-path construction the positions of team-pairs are fixed, we never arrange games between teams from two different team-pairs. So, the games between teams $x_i\in \overline{X_\varepsilon}$ and $y_j\in \overline{Y_\varepsilon}$ with $i\neq j$ are unarranged.
Moreover, in a left super-game two days' games in $m_0\cup\overline{m_0}=\{s_{i'}\leftrightarrow t_{i'}\}_{i'=0}^{3d}$ were not arranged, which include two games $\{s_{3d}\leftrightarrow t_{3d}\}$ between teams in the involved team-pair. And, all other games in $m_0\cup\overline{m_0}$ only involve teams in super-teams. So, the games between teams $x_i\in \overline{X_\varepsilon}$ and $y_j\in \overline{Y_\varepsilon}$ with $i=j$ are also unarranged. Therefore, the unarranged regular games (involving teams in team-pairs) are $\{x_i\leftrightarrow y_j\mid x_i\in \overline{X_\varepsilon}, y_j\in\overline{Y_\varepsilon}\}$.
\end{proof}

The games in $\{x_i\leftrightarrow y_j\mid x_i\in \overline{X_\varepsilon}, y_j\in\overline{Y_\varepsilon}\}$ forms a bipartite tournament between two $l$-team leagues $\overline{X_\varepsilon}$ and $\overline{Y_\varepsilon}$. Next, we design a simple algorithm to arrange them.

Let $s_i\coloneqq x_{n_\varepsilon+i}$ and $t_i\coloneqq y_{n_\varepsilon+i}$ for any $0\leq i<l$. We define a set of games as
\[
m_i\coloneqq\{s_{i'}\rightarrow t_{(i'+i)\bmod l}\}_{i'=0}^{l-1}.
\]
The games on $2l$ days for $\overline{X_\varepsilon}\cup\overline{Y_\varepsilon}$ can be presented by
\[
m_0\overline{m_1}m_2\overline{m_3}\cdots  \cdot\overline{m_0\overline{m_1}m_2\overline{m_3}\cdots}.
\]
Since $l\geq 3d$, we can avoid the infeasible case $m_0\overline{m_0}$ ($l=1$).

\textbf{The teams in super-teams.} 
By Claim~\ref{unique}, for any $1\leq i\leq l$, each super-team plays only one $i$-th left super-game. Hence, there are $m$ $i$-th left super-games, denoted as $\{S_i-T_{j_i}\}_{i=1}^{m}$. Recall that in the $i$-th left super-game two days' games in $m_0\cup\overline{m_0}$ were not arranged. We denote the union of unarranged games (involving teams in super-teams) in $m_i$ for all $m$ $i$-th left super-games as $M_i$. Then, the unarranged games (involving teams in super-teams) due to the left super-games are $\bigcup_{i=1}^{l}M_i\cup\overline{M_i}$. 
The games on $2l$ days for $X_\varepsilon\cup Y_\varepsilon$ can be presented by
\[
M_1\overline{M_2}M_3\overline{M_4}\cdots  \cdot\overline{M_1\overline{M_2}M_3\overline{M_4}\cdots}.
\]

\begin{theorem}\label{feasibility}
For BTTP with $n\geq 108d^3$, the above 3-path construction generates a feasible solution.
\end{theorem}
\begin{proof}
First, we need to prove that in our schedule all games $\{x\leftrightarrow y\mid x\in X, y\in Y\}$ are arranged once. Since our schedule spans $2n$ days that contain exactly $2n^2$ games, it is sufficient to show that all games $\{x\leftrightarrow y\mid x\in X, y\in Y\}$ are arranged.
For any two vertex sets $X'\subseteq X$ and $Y'\subseteq Y$, we use $(X', Y')$ to denote all games between one regular team in $X'$ and one regular team of $Y'$ for the sake of presentation.
In our 3-path construction there are $2m$ super-teams $\{S_1,\dots,S_m\}\cup \{T_1,\dots,T_m\}$ and $l$ team-pairs $\{L_1,\dots,L_l\}$.
\begin{enumerate}
    \item For two super-teams $S_i$ and $T_j$, the games $(S_i,T_j)$ need to be arranged. By Claim~\ref{onesuper}, there is one super-game between them. If it is a normal super-game, the games are arranged in extending the super-game. If it is a left super-game, the games are arranged except for two days' games in extending the super-game, where the unarranged two days' games belonging to $\bigcup_{i=1}^{l}M_i\cup\overline{M_i}$ are arranged in the last time slot.
    \item For one team-pair $L_{i'}=\{x_i,y_i\}$, the games $\{x_i\leftrightarrow y_i\}$ belonging to $(\overline{X_\varepsilon}, \overline{Y_\varepsilon})$ are arranged in the last time slot. 
    Moreover, for two team-pair $L_{i'}=\{x_i,y_i\}$ and $L_{j'}=\{x_j,y_j\}$, the games $\{x_i\leftrightarrow y_j\}\cup\{x_j\leftrightarrow y_i\}$ belonging to $(\overline{X_\varepsilon}, \overline{Y_\varepsilon})$ are also arranged in the last time slot.
    \item For one super-team $S_i$ (resp., $T_i$) and one team-pair $L_{j'}=\{x_j,y_j\}$, the games $(S_i, \{y_{j}\})$ (resp., $(T_i,\{x_{j}\})$) need to be arranged. By Claim~\ref{unique}, $S_i$ (resp., $T_i$) plays one $i'$-th left super-game, and the games are arranged in extending the left super-game.
\end{enumerate}

Then, we need to prove that the double round-robin bipartite tournament satisfies the no-repeat constraint. For any pair of teams $x\in X$ and $y\in Y$, the two games between them are arranged entirely in one super-game or the unarranged games involving teams in team-pairs or the unarranged games involving teams in super-teams.
It is easy to check that the design of each satisfies the constraint.

Last, we need to prove that the double round-robin bipartite tournament satisfies the bounded-by-3 constraint.
We may only consider a team $x_i\in X$ since the analysis is similar for a team $y_i\in Y$. 
\begin{enumerate}
    \item Assume $i\geq n_\varepsilon$. So, $x_i$ belongs to a team-pair, denoted as $L_{i'}$. In each of the first $m$ time slots, $L_{i'}$ always plays a left super-game, where $x_i$ plays an away game and a home game alternatively on $6d$ days (see Table~\ref{left-example}). In the last time slot, $x_i$ plays an away game and a home game on the first two days, and then it plays at most 2-consecutive home/away games on the rest $l-2$ days (note that $l\geq 3d>3$). Hence, $x_i$ plays at most 3-consecutive home/away in our schedule.
    \item Assume $i<n_\varepsilon$. So, $x_i$ belongs to a super-team, denoted as $S_{i'}$. Let $t\coloneqq i\bmod 3$. In each of the first $m$ time slots, $S_{i'}$ either plays a left super-game or a normal super-game. If it plays a left super-game, $x_i$ plays an away game and a home game alternatively on $6d$ days (see Table~\ref{left-example}); if it plays a normal super-game,  on $6d$ days, $x_i$ firstly plays $t$-consecutive home games and 3-consecutive away games, and then plays 3-consecutive home games and 3-consecutive away games alternatively, and lastly plays $(3-t)$-consecutive home games (see Table~\ref{normal-example}). In the last time slot, $x_i$ plays an away game and a home game on the first two days, and then it plays at most 2-consecutive home/away games on the rest $l-2$ days. Hence, $x_i$ plays at most 3-consecutive home/away in our schedule.
\end{enumerate}

Therefore, our 3-path construction outputs a feasible schedule. Note that it requires $n\geq 108d^3$ by the setting so that $l\leq m$.
\end{proof}

Although the above theorem requires $n$ being a large number, it can guarantee our expected approximation ratio since an instance with a constant $n$ can be optimally solved in constant time.
However, for the sake of application, we also show that our 3-path construction can be slightly modified to work for instances with almost all $n$.

\begin{theorem}\label{feasibility+}
For BTTP, our 3-path construction can be slightly modified to form a new 3-path construction, which can generate a feasible solution for any $n$ except for $n\in\{1,2,5,8,14\}$.
\end{theorem}
\begin{proof}
Recall that $d$ is the number of 3-paths with respect to one super-team, $m$ is the number of super-teams in one league, and $l$ is the the number of team-pairs. We have $n_\varepsilon=3md$ and $n=n_\varepsilon+l$. Previously, we set $d\coloneqq6\lceil 1/\varepsilon \rceil$ and $m\coloneqq2\floor{\frac{n}{6d}}-1$, which is not necessary.

We present some necessary conditions of our previous 3-path construction.
Firstly, it requires that $m\geq l$ since there are $m$ super-games, and among them, there are $l$ left super-games.
Secondly, it requires that $m$ is odd; otherwise, Claim~\ref{unique} will not hold.
Lastly, it requires that $l\neq 1$; otherwise, in the time slot games for teams in the team-pair (resp., super-teams) will be arranged as $m_0\overline{m_0}$ (resp., $M_1\overline{M_1}$) on the last two days, violating the no-repeat constraint. 

Note that the problem that $l=1$ can be fixed as follows: we rearrange the games $m_0\cup\overline{m_0}\cup M_1\cup\overline{M_1}$ in the last time slot as follows: we arrange games in $\overline{m_0}\cup\overline{M_1}$ on the first day, and move the previous games on the $i$-th to the $(i+1)$-th day for each $1\leq i<2n$. 
That is, if we denote the previous games on $2n$ days in our previous 3-path construction as $g_1g_2\dots g_{2n}$, we rearrange them in the order of $g_{2n}g_1g_2\dots g_{2n-1}$, which can generate a feasible solution by a similar proof of Theorem~\ref{feasibility}.

Hence, we only need to make sure that $m$ is odd and $m\geq l$. Note that $n_\varepsilon=3md\leq n$ and $l=n-n_\varepsilon=n-3md\leq m$. Therefore, $m$ and $d$ should satisfy that
\begin{itemize}
    \item $n\geq 3md$;
    \item $m\geq n-3md$;
    \item $m$ is odd.
\end{itemize}
For almost any $n$, we can find feasible $m$ and $d$.
For example, if $n\equiv 0\pmod 3$, we simply set $d=n/3$ and $m=1$, and in this case our new 3-path construction becomes the same as in~\cite{hoshino2013approximation}.
Moreover, by setting $d=1$ and $m=2\floor{\frac{n}{6}}-1$, it is easy to see that for any $n\geq 24$ the constrains are also satisfied. However, by testing every $n$ satisfying $1\leq n<24$, we find that no such $m$ and $d$ exist only for $n\in\{1,2,5,8,14\}$.

Therefore, our 3-path construction can be slightly modified to form a new 3-path construction, which can generate a feasible solution for any $n$ except for $n\in\{1,2,5,8,14\}$.
\end{proof}

\begin{remark}
In Theorem~\ref{feasibility+}, the new 3-path construction relies on ``greedily chosen'' parameters $m$ and $d$. While this approach generates a feasible solution for almost all $n$, it may only be used to design a heuristic algorithm, as $d$ may be too small and then it cannot ensure that almost all teams play 3 consecutive away games along nearly all 3-paths in a given 3-path packing.
\end{remark}

\section{The Analysis}\label{SC.4}
In this section, we first introduce a new lower bound for any optimal solution to BTTP, which is related to the minimum weight cycle packing. Then, we analyze the quality of our 3-path construction, which reveals the relation between the quality of the obtained schedule and the used 3-path packing. At last, we show how to use the minimum weight cycle packing to compute a good 3-path packing such that by applying this 3-path packing to our 3-path construction, we can get a schedule with an approximation ratio of at most $(3/2+\varepsilon)$ in polynomial time.

\subsection{Lower Bounds}\label{SC.4.2}
Lower bounds play an important role in approximation algorithms. We need to compare our solution with an optimal solution. However, it is hard to compute an optimal solution. Then, we turn to find some lower bounds of the optimal value 
and compare our solution with the lower bounds. We will use two lower bounds. The first one is from the literature and the second one is newly proved.

\begin{lemma}[\cite{hoshino2013approximation}]\label{delta}\label{lb1}
For BTTP, it holds that $2\delta_X(Y)=\delta_Y(X)+\delta_X(Y)\leq \frac{3}{2}\cdot\mbox{OPT}$.
\end{lemma}

Next, we propose a new lower bound that is related to the minimum weight cycle packings $\C_{X_\varepsilon}$ and $\C_{Y_\varepsilon}$.
\begin{lemma}\label{lb2}
For BTTP, it holds that $\delta_{Y_\varepsilon}(X_\varepsilon)+\frac{1}{2}n_\varepsilon w(\C_{Y_\varepsilon})+\delta_{X_\varepsilon}(Y_\varepsilon)+\frac{1}{2}n_\varepsilon w(\C_{X_\varepsilon})\leq \frac{3}{2}\cdot\OPT$.
\end{lemma}
\begin{proof}
Fix an optimal solution of BTTP. For every team $v\in X\cup Y$, we use $\OPT(v)$ to denote the weight of its itinerary in the optimal solution. So, $\OPT=\sum_{v\in X\cup Y}\OPT(v)$. 
It is sufficient to prove $\delta_{Y_\varepsilon}(x)+\frac{1}{2}w(\C_{Y_\varepsilon})\leq \frac{3}{2}\cdot\OPT(x)$ and $\delta_{X_\varepsilon}(y)+\frac{1}{2}w(\C_{X_\varepsilon})\leq\frac{3}{2}\cdot\OPT(y)$ for any $x\in X_\varepsilon$ and $y\in Y_\varepsilon$. In the following, we only prove the first inequality since the second can be obtained similarly.


For a team $x\in X$, we decompose its itinerary in the optimal solution into a set of cycles $\C=\C_4\cup\C_3\cup\C_2$, where $\C_i$ is a set of $i$-cycles containing $x$. The cycles cover $x$ and all vertices in $Y$. We can shortcut vertices in $\overline{Y_\varepsilon}$ to obtain a new set of cycles $\C'=\C'_4\cup\C'_3\cup\C'_2$ that cover $x$ and all vertices in $Y_\varepsilon$ only, where $\C'_i$ is a new set of $i$-cycles containing $x$. By the triangle inequality, we have for every $x$
\[
w(\C')\leq w(\C)=\OPT(x).
\]

Since cycles in $\C$ share only one common vertex $x$, cycles in $\C'$ also share only one common vertex $x$. Moreover, since cycles in $\C'$ cover all vertices in $Y_\varepsilon$, we can get
\begin{equation}\label{new1}
\begin{aligned}
\delta_{Y_\varepsilon}(x)&=\sum_{{xy_1y_2y_3x \in\C'_4}}(w(xy_1)+w(xy_2)+w(xy_3))+\sum_{xy_1y_2x\in\C'_3}(w(xy_1)+w(xy_2))+\sum_{xy_1x\in\C'_2}w(xy_1)\\
&\leq\sum_{{xy_1y_2y_3x \in\C'_4}}(w(xy_1)+w(xy_3))+\frac{1}{2}w(\C'_4)+w(\C'_3)+\frac{1}{2}w(\C'_2)\\
&\leq\sum_{{xy_1y_2y_3x \in\C'_4}}(w(xy_1)+w(xy_3))+\frac{1}{2}w(\C'_4)+w(\C'_3)+w(\C'_2),
\end{aligned}
\end{equation}
where the first inequality is due to the following facts: for a 4-cycle $xy_1y_2y_3x\in\C'_4$, we have $w(xy_2)\leq \frac{1}{2}w(xy_1y_2y_3x)$ by the triangle inequality; for a 3-cycle $xy_1y_2y_3x\in\C'_3$, we have $w(xy_1)+w(xy_2)\leq w(xy_1y_2x)$; and for a 2-cycle $xy_1x\in\C'_2$ we have $w(xy_1)=\frac{1}{2}w(xy_1x)$.


Let $n'\coloneqq\size{\C'_4}$. 
Then, cycles in $\C'_4$ cover $3n'$ vertices in $Y_\varepsilon$, and cycles in $\C'_3\cup\C'_2$ cover $n_\varepsilon-3n'$ vertices in $Y_\varepsilon$.
Recall that cycles in $\C'_3\cup\C'_2$ share only one common vertex $x$. Then, we can get a single cycle $C$ by shortcutting $x$ that covers $n_\varepsilon-3n'$ vertices in $Y_\varepsilon$ and satisfies $w(C)\leq w(\C'_3\cup\C'_2)$ by the triangle inequality. 
Let $\C''=\bigcup_{xy_1y_2y_3\in\C'_4}\{y_1y_2y_3y_1\}\cup\{C\}$.
Since $n_\varepsilon$ is divisible by 3, we know that $n_\varepsilon-3n'$ is also divisible by 3. Thus, if $n_\varepsilon-3n'>0$, the order of $C$ is at least 3. Therefore, $\C''$ is a cycle packing in $G[Y_\varepsilon]$. Since $\C_{Y_\varepsilon}$ is a minimum weight cycle packing in $G[Y_\varepsilon]$, we have $w(\C_{Y_\varepsilon})\leq w(\C'')$. Hence, we have
\begin{equation}\label{new2}
w(\C_{Y_\varepsilon})\leq\sum_{{xy_1y_2y_3x \in\C'_4}}w(y_1y_2y_3y_1)+w(C)\leq\sum_{{xy_1y_2y_3x \in\C'_4}}2w(y_1y_2y_3)+w(\C'_3)+w(\C'_2),
\end{equation}
where the inequality is due to the fact that $w(y_1y_3)\leq w(y_1y_2y_3)$ by the triangle inequality, and $w(C)\leq w(\C'_3\cup\C'_2)=w(\C'_3)+w(\C'_2)$.

Therefore, by (\ref{new1}) and (\ref{new2}), we have
\begin{align*}
\delta_{Y_\varepsilon}(x)+\frac{1}{2}w(\C_{Y_\varepsilon})&\leq\sum_{{xy_1y_2y_3x \in\C'_4}}(w(xy_1)+w(xy_3))+\frac{1}{2}w(\C'_4)+w(\C'_3)+w(\C'_2)\\
&\quad\quad+\sum_{{xy_1y_2y_3x \in\C'_4}}w(y_1y_2y_3)+\frac{1}{2}w(\C'_3)+\frac{1}{2}w(\C'_2)\\
&=w(\C'_4)+\frac{1}{2}w(\C'_4)+\frac{3}{2}w(\C'_3)+\frac{3}{2}w(\C'_2)\\
&=\frac{3}{2}w(\C')\leq\frac{3}{2}w(\C)=\frac{3}{2}\cdot\OPT(x).
\end{align*}
\end{proof}

\subsection{Upper Bounds}\label{SC.4.1}
The 3-path construction provides a feasible solution if each team in $X\cup Y$ has a label. There are $n!\times n!$ ways to label teams in $X$ using $\{x_0,\dots,x_{n-1}\}$ and teams in $Y$ using $\{y_0,\dots,y_{n-1}\}$.
To get a good solution, we label them randomly. Note that in random operations, we operate uniformly without distinction.





\medskip
\noindent\textbf{Step~1.} Assign each vertex $x \in X$ (resp., $y\in Y$) with a cost of $\delta_Y(x)$ (resp., $\delta_X(y)$), and select the top $n_\varepsilon$ vertices with the highest costs to form $X_\varepsilon$ (resp., $Y_\varepsilon$).



\noindent\textbf{Step~2.} Generate a random permutation of the $md$ 3-paths in $\P_{X_\varepsilon}$ (resp., $\P_{Y_\varepsilon}$), and label them as $P_1,\dots,P_{md}$ (resp., $Q_1,\dots,Q_{md}$), respectively, where $P_i=x_{3i-3}x_{3i-2}x_{3i-1}$ and $Q_i=y_{3i-3}y_{3i-2}y_{3i-1}$.

\noindent\textbf{Step~3.} Take an arbitrary permutation of $l$ teams in $\overline{X_\varepsilon}$ (resp., $\overline{Y_\varepsilon}$), and label them as $x_{n_\varepsilon},\dots,x_{n-1}$ (resp., $y_{n_\varepsilon},\dots,y_{n-1}$), respectively.

\medskip

Note that {Step~1} can be done in $O(n^2)$ time, and {Step~2} and {Step~3} can be done in $O(n)$ time by using the Fisher–Yates shuffle~\cite{knuth1997art}. In {Step~2}, we assume that the two 3-path packings are given in advance. Moreover, when labeling a 3-path using $P_i=x_{3i-3}x_{3i-2}x_{3i-1}$, there are two choices based on the two opposite directions, and either of them is considered valid.

\begin{theorem}\label{res-construction}
For BTTP with $n\geq 108d^3$, given any 3-path packing $\P_{X_\varepsilon}$ (resp., $\P_{Y_\varepsilon}$) in $G[X_\varepsilon]$ (resp., $G[Y_\varepsilon]$), using the above randomized labeling method, our 3-path construction can generate a solution with an expected weight of at most $\delta_{Y_\varepsilon}(\P_{X_\varepsilon})+n_\varepsilon w(\P_{X_\varepsilon})+\delta_{X_\varepsilon}(\P_{Y_\varepsilon})+n_\varepsilon w(\P_{Y_\varepsilon})+\varepsilon\cdot\OPT$ in $O(n^2)$ time.
\end{theorem}
\begin{proof}
Clearly, our 3-path construction takes $O(n^2)$ time. 

Then, we make two assumptions for the sake of analysis, which do not decrease the weight of our schedule by the triangle inequality.
Firstly, we assume that every participant team returns home both before and after each day's game in left super-games and the last time slot. Secondly, we assume that all participant teams return home before the game on the first day and after the game on the last day in normal super-games.

By {Step~1} of our randomized labeling algorithm, since we select the top $n_\varepsilon$ vertices from $X$ with the highest costs to form $X_\varepsilon$, we get 
\begin{equation*}
\delta_Y(X_\varepsilon)=\sum_{x\in X_\varepsilon}\delta_Y(x)\geq \frac{n_\varepsilon}{n}\sum_{x\in X}\delta_Y(x)=\frac{n_\varepsilon}{n}\delta_Y(X).
\end{equation*}
Alternatively, since $l=n-n_\varepsilon\leq 9d$ by (\ref{eq0}), we have
\begin{equation}\label{eq1}
\delta_{\overline{X_\varepsilon}}(Y)=\delta_Y(\overline{X_\varepsilon})\leq \lrA{1-\frac{n_\varepsilon}{n}}\delta_Y(X)\leq\frac{9d}{n}\delta_Y(X).
\end{equation}
Similarly, we have 
\begin{equation}\label{eq2}
\delta_{\overline{Y_\varepsilon}}(X)=\delta_X(\overline{Y_\varepsilon})\leq \lrA{1-\frac{n_\varepsilon}{n}}\delta_X(Y)\leq\frac{9d}{n}\delta_X(Y).
\end{equation}


Consider a team $x\in \overline{X_\varepsilon}$.
By assumptions, in our schedule the weight of its itinerary is $\sum_{y\in Y}2w(x,y)=2\delta_Y(x)$. The total weight of itineraries of teams in $\overline{X_\varepsilon}$ is $\sum_{x\in \overline{X_\varepsilon}}2\delta_Y(x)=2\delta_Y(\overline{X_\varepsilon})\leq \frac{18d}{n}\delta_Y(X)\leq\frac{1}{d}\delta_Y(X)$ by (\ref{eq1}) and $n\geq 108d^3$.

Fix a team $x\in S_i\subseteq X_\varepsilon$. We have three parts of its trips.
\begin{enumerate}
    \item The weight of its trips for playing teams in $\overline{Y_\varepsilon}$ is $\sum_{y\in \overline{Y_\varepsilon}}2w(x,y)=2\delta_{\overline{Y_\varepsilon}}(x)$ by assumptions since these games are in left super-games.
    \item In a left super-game between $S_i$ and $T_j$, the weight of its trips for playing teams in $T_j$ is $\sum_{y\in T_j}2w(x,y)=2\delta_{T_j}(x)$. We can get $\EE{\delta_{T_j}(x)}=\frac{1}{m}\delta_{Y_\varepsilon}(x)$ by {Step~2} of our randomized labeling algorithm. Since $S_i$ plays $l$ left super-games in total, the expected weight of its trips for playing teams in a super-team that plays a left super-game with $S_i$ is $\frac{2l}{m}\delta_{Y_\varepsilon}(x)\leq \frac{1}{d}\delta_{Y_\varepsilon}(x)$ due to $l\leq 9d$ and $m\geq18d^2$ by (\ref{eq0}).
    \item In a normal super-game between $S_i$ and $T_j$, it plays $d$ or $d+1$ trips along at least $d-1$ 3-paths in $\Q_j\coloneqq\{Q_{(j-1)d+1},\dots,Q_{jd}\}$ from one terminal to another (see Table~\ref{normal-example}). If it plays $d$ trips along $d$ 3-paths in $\Q_j$, the weight of these trips is $\delta_{x}(\Q_j)+w(\Q_j)$; if it plays $d+1$ trips along $d-1$ 3-paths in $\Q_j$, it does not follow a 3-path $y'_0y'_1y'_2$ in $\Q_j$, the weight of the two trips for playing teams in $\{y'_0,y'_1,y'_2\}$ is bounded by $2w(x,y'_0)+2w(x,y'_1)+2w(x,y'_2)$ by the triangle inequality (Note that, alternatively, here we can assume that each team returns home both before and after each day's game during trips where they visit only two teams.), and then the weight of the $d+1$ trips is bounded by $\delta_{x}(\Q_j)+w(\Q_j)+2w(x,y'_0)+2w(x,y'_1)+2w(x,y'_2)$. We have $\EE{\delta_{x}(\Q_j)}=\frac{1}{m}\delta_{x}(\P_{Y_\varepsilon})$, $\EE{w(\Q_j)}=\frac{1}{m}w(\P_{Y_\varepsilon})$, and $\EE{w(x,y'_0)+w(x,y'_1)+w(x,y'_2)}=\frac{1}{md}\delta_{Y_\varepsilon}(x)$ by {Step~2} of our randomized labeling algorithm. Since $S_i$ plays $m-l\leq m$ normal super-games, the expected weight of trips for playing teams in a super-team that plays a normal super-game with $S_i$ is at most $\delta_{x}(\P_{Y_\varepsilon})+w(\P_{Y_\varepsilon})+\frac{1}{d}\delta_{Y_\varepsilon}(x)$.
\end{enumerate}
So, the expected weight of the itinerary of team $x$ is at most $2\delta_{\overline{Y_\varepsilon}}(x)+\frac{2}{d}\delta_{Y_\varepsilon}(x)+\delta_{x}(\P_{Y_\varepsilon})+ w(\P_{Y_\varepsilon})$.
Hence, the total expected weight of itineraries of teams in $X_\varepsilon$ is at most 
\begin{align*}
&2\delta_{\overline{Y_\varepsilon}}(X_\varepsilon)+\frac{2}{d}\delta_{Y_\varepsilon}(X_\varepsilon)+\delta_{X_\varepsilon}(\P_{Y_\varepsilon})+n_\varepsilon w(\P_{Y_\varepsilon})\leq \frac{3}{d}\delta_X(Y)+\delta_{X_\varepsilon}(\P_{Y_\varepsilon})+n_\varepsilon w(\P_{Y_\varepsilon})
\end{align*}
since we have $2\delta_{\overline{Y_\varepsilon}}(X_\varepsilon)\leq 2\delta_{\overline{Y_\varepsilon}}(X)\leq \frac{1}{d}\delta_{X}(Y)$ by (\ref{eq2}) and $n\geq 108d^3$, and $\delta_{Y_\varepsilon}(X_\varepsilon)\leq \delta_{Y}(X)=\delta_{X}(Y)$.

Therefore, the total expected weight of itineraries of teams in $X$ is bounded by $\frac{4}{d}\delta_X(Y)+\delta_{X_\varepsilon}(\P_{Y_\varepsilon})+n_\varepsilon w(\P_{Y_\varepsilon})$. 
Similarly, for teams in $Y$, the expected weight is bounded by $\frac{4}{d}\delta_Y(X)+\delta_{Y_\varepsilon}(\P_{X_\varepsilon})+n_\varepsilon w(\P_{X_\varepsilon})$.

Since $d=6\ceil{1/\varepsilon}\geq6/\varepsilon$ by our setting and $\delta_Y(X)=\delta_X(Y)\leq (3/4)\cdot\OPT$ by Lemma~\ref{lb1}, our schedule has an expected weight of at most $\varepsilon\cdot\OPT+\delta_{Y_\varepsilon}(\P_{X_\varepsilon})+n_\varepsilon w(\P_{X_\varepsilon})+\delta_{X_\varepsilon}(\P_{Y_\varepsilon})+n_\varepsilon w(\P_{Y_\varepsilon})$.
\end{proof}


\subsection{The 3-Path Packing}\label{SC.4.3}
In this section, we will obtain novel 3-path packings $\P_{X_\varepsilon}$ in $G[X_\varepsilon]$ and $\P_{Y_\varepsilon}$ in $G[Y_\varepsilon]$ such that 
\[
\delta_{Y_\varepsilon}(\P_{X_\varepsilon})+n_\varepsilon w(\P_{X_\varepsilon})+\delta_{X_\varepsilon}(\P_{Y_\varepsilon})+n_\varepsilon w(\P_{Y_\varepsilon})\leq \delta_{Y_\varepsilon}(X_\varepsilon)+\frac{1}{2}n_\varepsilon w(\C_{X_\varepsilon})+\delta_{X_\varepsilon}(Y_\varepsilon)+\frac{1}{2}n_\varepsilon w(\C_{Y_\varepsilon}).
\]
By Lemma~\ref{lb2} and Theorem~\ref{res-construction}, these two 3-path packings directly imply a $(3/2+\varepsilon)$-approximation for BTTP.

Next, we only show how to obtain $\P_{X_\varepsilon}$ in $G[X_\varepsilon]$ such that $\delta_{Y_\varepsilon}(\P_{X_\varepsilon})+n_\varepsilon w(\P_{X_\varepsilon})\leq \delta_{Y_\varepsilon}(X_\varepsilon)+\frac{1}{2}n_\varepsilon w(\C_{X_\varepsilon})$ since we can get $\P_{Y_\varepsilon}$ in a similar way. We define some notations.

Since $\C_{X_\varepsilon}$ is a cycle packing in $G[X_\varepsilon]$, the order of each cycle in $\C_{X_\varepsilon}$ is at least 3 and at most $n_\varepsilon$. Let $\B_q$ be the set of all $q$-cycles in $\C_{X_\varepsilon}$. So, we have $\C_{X_\varepsilon}=\bigcup_{q=3}^{n_\varepsilon}\B_q$.

We will use $\C_{X_\varepsilon}$ to obtain a cycle packing $\C$ in $G[X_\varepsilon]$ such that the order of each cycle in $\C$ is divisible by 3, and then use $\C$ to obtain the 3-path packing $\P_{X_\varepsilon}$. The algorithm for $\C$ is shown as follows.

\medskip
\noindent\textbf{Step~1.} Initialize $\C=\emptyset$.

\noindent\textbf{Step~2.} For each 3-cycle in $\B_3$, we directly choose it into $\C$. 

\noindent\textbf{Step~3.} For each $q$-cycle $C_q\in\C_{X_\varepsilon}\setminus\B_3$, obtain a $q$-path $P_q$ by deleting the edge $xx'\in C_q$ such that $\delta_{Y_\varepsilon}(x)+\delta_{Y_\varepsilon}(x')-n_\varepsilon w(x,x')$ is minimized. Then, patch all paths into a single cycle $C$ arbitrarily, and select the cycle into the packing.
\medskip

An illustration of the patching operation in {Step 3} can be seen in Figure~\ref{fig03}.
\begin{figure}[ht]
\centering
\begin{tikzpicture}[scale=0.9]
\draw[very thick] (-2,0.5).. controls (-1.5,1.5) and (1.5,1.5) ..(2,0.5);
\draw[very thick] (-2,-0.5).. controls (-1.5,-1) ..(-0.5,-1);
\draw[very thick] (2,-0.5).. controls (1.5,-1) ..(0.5,-1);
\draw[very thick,dotted] (2,-0.5) to (2,0.5);
\draw[very thick,dotted] (-2,-0.5) to (-2,0.5);
\draw[very thick,dotted] (-0.5,-1) to (0.5,-1);

\node[left] at (-2,0) {\small $e_1$};
\node[right] at (2,0) {\small $e_3$};
\node[below] at (0,-1) {\small $e_2$};

\filldraw [black]
(2,-0.5) circle [radius=2pt]
(2,0.5) circle [radius=2pt]
(-2,-0.5) circle [radius=2pt]
(-2,0.5) circle [radius=2pt]
(-0.5,-1) circle [radius=2pt]
(0.5,-1) circle [radius=2pt];
\end{tikzpicture}
\caption{An illustration of patching paths into a cycle: there are three paths, denoted by solid lines, and they are patched into a cycle using three edges $e_1,e_2,e_3$, denoted by dashed lines.}
\label{fig03}
\end{figure}
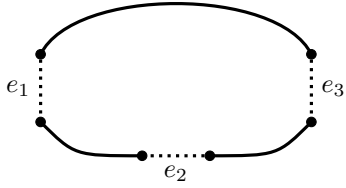
Moreover, it is easy to see that the order of each cycle in $\C$ is divisible by 3, and the algorithm takes $O(n^3)$ time, dominated by computing the minimum weight cycle packing $\C_{X_\varepsilon}$~\cite{hartvigsen1984extensions}.

\begin{lemma}\label{f1}
$n_\varepsilon w(\C)\leq\frac{1}{2}\delta_{Y_\varepsilon}(X_\varepsilon)+\frac{3}{4}n_\varepsilon w(\C_{X_\varepsilon})$.
\end{lemma}
\begin{proof}
We define some notations.

We may consider that each vertex $x\in X_\varepsilon$ has a cost of $\delta_{Y_\varepsilon}(x)$. 
Let $\mu_q=\sum_{C\in\B_q}\sum_{x\in C}\delta_{Y_\varepsilon}(x)$ and $\lambda_q=n_\varepsilon w(\B_q)$, where $\mu_q$ is the total cost of vertices that appear in a $q$-cycle in $\B_q$. Then, we have
\begin{equation}\label{eq3}
\delta_{Y_\varepsilon}(X_\varepsilon)=\sum_{3\leq q\leq n_\varepsilon}\mu_q~~~\mbox{and}~~~n_\varepsilon w(\C_{X_\varepsilon})=\sum_{3\leq q\leq n_\varepsilon}\lambda_q.
\end{equation}
For any $x,x'\in X_\varepsilon$, by the triangle inequality, we have 
\begin{equation}\label{degree}
n_\varepsilon w(x,x')\leq\sum_{y\in Y_\varepsilon}(w(x,y)+w(x',y))=\delta_{Y_\varepsilon}(x)+\delta_{Y_\varepsilon}(x').
\end{equation}
By (\ref{degree}), for any $C\in\C_{X_\varepsilon}$, $n_\varepsilon w(C)=\sum_{xx'\in C}n_\varepsilon w(x,x')\leq\sum_{xx'\in C}(\delta_{Y_\varepsilon}(x)+\delta_{Y_\varepsilon}(x'))=2\sum_{x\in C}\delta_{Y_\varepsilon}(x)$. So, we have
\begin{align}\label{d<=2D}
\begin{split}
\lambda_q=n_\varepsilon w(\B_q)&=\sum_{C\in\B_q}n_\varepsilon w(C)\leq 2\sum_{C\in\B_q}\sum_{x\in C}\delta_{Y_\varepsilon}(x)=2\mu_q. 
\end{split}
\end{align}
Let $\mu_{>3}\coloneqq\sum_{q=4}^{n_\varepsilon}\mu_q$ and $\lambda_{>3}\coloneqq\sum_{q=4}^{n_\varepsilon}\lambda_q$. By (\ref{eq3}), we have
\begin{equation}\label{d=D}
\delta_{Y_\varepsilon}(X_\varepsilon)=\mu_3+\mu_{>3}\quad\mbox{and}\quad n_\varepsilon w(\C_{X_\varepsilon})=\lambda_3+\lambda_{>3}.
\end{equation}

Next, we are ready to prove the lemma.

By {Step 2} of the algorithm, the contributions of cycles in $\B_3$ on $n_\varepsilon w(\C)$ are $n_\varepsilon w(\B_3)=\lambda_3$. 


Then, we consider the cycles in $\C_{X_\varepsilon}\setminus\B_3$. By {Step 3}, the contributions of them on $n_\varepsilon w(\C)$ is 
$n_\varepsilon w(C)$, where $C$ is obtained by patching a set of paths, denoted by $\P$. The paths in $\P$ are patched into the cycle $C$ using a set of vertex-disjoint edges (see Figure~\ref{fig03}), denoted by $E'$.
By~(\ref{degree}), $\sum_{e\in E'}n_\varepsilon w(e)$ is bounded by the cost of terminals of paths in $\P$, i.e., $\delta_{Y_\varepsilon}(\P)$. Hence, we have 
\[
n_\varepsilon w(C)\leq n_\varepsilon w(\P)+\delta_{Y_\varepsilon}(\P).
\]
For each $q$-cycle in $\C_{X_\varepsilon}\setminus\B_3$, the contribution of the cycle on $n_\varepsilon w(\P)+\delta_{Y_\varepsilon}(\P)$ is $n_\varepsilon w(P_q)+\delta_{Y_\varepsilon}(P_q)$, where the $q$-path $P_q$ is obtained by deleting the edge $xx'\in C_q$ with $\delta_{Y_\varepsilon}(x)+\delta_{Y_\varepsilon}(x')-n_\varepsilon w(x,x')$ minimized by {Step 3}. Then, we have
\begin{align*}
\delta_{Y_\varepsilon}(x)+\delta_{Y_\varepsilon}(x')-n_\varepsilon w(x,x')&\leq \frac{1}{q}\sum_{x''x'''\in C_q}(\delta_{Y_\varepsilon}(x'')+\delta_{Y_\varepsilon}(x''')-n_\varepsilon w(x'',x'''))\\
&=\frac{2}{q}\sum_{x''\in C_q}\delta_{Y_\varepsilon}(x'')-\frac{1}{q}n_\varepsilon w(C_q).
\end{align*}
Hence, we have 
\begin{align*}
n_\varepsilon w(P_q)+\delta_{Y_\varepsilon}(P_q)&=n_\varepsilon w(C_q)+\delta_{Y_\varepsilon}(x)+\delta_{Y_\varepsilon}(x')-n_\varepsilon w(x,x')\\
&\leq \frac{2}{q}\sum_{x''\in C_q}\delta_{Y_\varepsilon}(x'')+\frac{q-1}{q}n_\varepsilon w(C_q)
\end{align*}
Then, we have
\begin{align*}
n_\varepsilon w(\P)+\delta_{Y_\varepsilon}(\P)&=\sum_{P_q\in\P}\lrA{n_\varepsilon w(P_q)+\delta_{Y_\varepsilon}(P_q)}\\
&\leq\sum_{C_q\in\C_{X_\varepsilon}\setminus\B_3}\lrA{\frac{2}{q}\sum_{x''\in C_q}\delta_{Y_\varepsilon}(x'')+\frac{q-1}{q}n_\varepsilon w(C_q)}\\
&=\sum_{q=4}^{n_\varepsilon}\sum_{C_q\in\B_q}\lrA{\frac{2}{q}\sum_{x''\in C_q}\delta_{Y_\varepsilon}(x'')+\frac{q-1}{q}n_\varepsilon w(C_q)}\\
&=\sum_{q=4}^{n_\varepsilon}\lrA{\frac{2}{q}\mu_q+\frac{q-1}{q}\lambda_q}\\
&\leq \sum_{q=4}^{n_\varepsilon}\lrA{\frac{2}{4}\mu_q+\frac{3}{4}\lambda_q}=\frac{1}{2}\mu_{>3}+\frac{3}{4}\lambda_{>3},
\end{align*}
where the last inequality follows from $\frac{2}{q}\mu_q+\frac{q-1}{q}\lambda_q\leq \frac{2}{4}\mu_q+\frac{3}{4}\lambda_q$ for any $q\geq 4$ since $\lambda_q\leq2\mu_q$ for any $q\geq 4$ by (\ref{d<=2D}). 

Therefore, we have 
\begin{align*}
n_\varepsilon w(\C)&=n_\varepsilon w(\B_3)+n_\varepsilon w(C)\\
&\leq\lambda_3+\frac{1}{2}\mu_{>3}+\frac{3}{4}\lambda_{>3}\\
&\leq \frac{1}{2}\mu_3+\frac{1}{2}\mu_{>3}+\frac{3}{4}\lambda_3+\frac{3}{4}\lambda_{>3}\\
&=\frac{1}{2}\delta_{Y_\varepsilon}(X_\varepsilon)+\frac{3}{4}n_\varepsilon w(\C_{X_\varepsilon}),
\end{align*}
where the last inequality follows from $\lambda_3\leq \frac{1}{2}\mu_3+\frac{3}{4}\lambda_3$ since $\lambda_3\leq2\mu_3$ by (\ref{d<=2D}), and the last equality follows from (\ref{d=D}).
\end{proof}

Next, we are ready to demonstrate the algorithm for $\P_{X_\varepsilon}$.

\medskip
\noindent\textbf{Step~1.} Compute in $O(n^3)$ time the cycle packing $\C$ in $G[X_\varepsilon]$ using the algorithm before Lemma~\ref{f1}.

\noindent\textbf{Step~2.} Initialize $\P_{X_\varepsilon}=\emptyset$.

\noindent\textbf{Step~3.} For each cycle $C\in\C$, orient it in an arbitrary direction, obtain a set of 3-paths $\P_C$ by deleting an edge every 2 edges along the oriented cycle starting from an initial vertex of $C$ to minimize $\delta_{Y_\varepsilon}(\P_C)+n_\varepsilon w(\P_C)$, and update $\P_{X_\varepsilon}\coloneqq \P_{X_\varepsilon}\cup \P_C$.
\medskip

Note that in Step 3, for each cycle $C\in\C$, there are $\size{C}$ ways to choose the initial vertex. By enumerating all $\size{C}$ cases and selecting the best one, we can find $\P_C$ in $O(\size{C}^2)$ time such that $\delta_{Y_\varepsilon}(\P_C)+n_\varepsilon w(\P_C)$ is minimized. 
This ensures that Step 3 takes $O(n^2)$ time. In fact, since $\size{C}\bmod 3=0$, it is sufficient to consider an initial vertex from only 3 consecutive vertices of $C$, reducing the number of cases to just 3.

Therefore, the algorithm for $\P_{X_\varepsilon}$ takes $O(n^3)$ time, dominated by computing $\C$ in Step 1.

\begin{theorem}\label{3-path}
The above algorithm takes $O(n^3)$ time to get a 3-path packing $\P_{X_\varepsilon}$ in $G[X_\varepsilon]$ such that $\delta_{Y_\varepsilon}(\P_{X_\varepsilon})+n_\varepsilon w(\P_{X_\varepsilon})\leq \delta_{Y_\varepsilon}(X_\varepsilon)+\frac{1}{2}n_\varepsilon w(\C_{X_\varepsilon})$. Similarly, there exists a $O(n^3)$-time algorithm to get a 3-path packing $\P_{Y_\varepsilon}$ in $G[Y_\varepsilon]$ with $\delta_{X_\varepsilon}(\P_{Y_\varepsilon})+n_\varepsilon w(\P_{Y_\varepsilon})\leq \delta_{X_\varepsilon}(Y_\varepsilon)+\frac{1}{2}n_\varepsilon w(\C_{Y_\varepsilon})$.
\end{theorem}
\begin{proof}
Consider the algorithm for $\P_{X_\varepsilon}$. As mentioned, its running time is $O(n^3)$. Then, we get an upper bound of $\delta_{Y_\varepsilon}(\P_{X_\varepsilon})+n_\varepsilon w(\P_{X_\varepsilon})$.

In Step 3, for each cycle $C\in\C$, if we choose the initial vertex uniformly at random, we can get a set of 3-paths $\P'_C$ such that $\EE{\delta_{Y_\varepsilon}(\P'_{C})}=\frac{2}{3}\sum_{x\in C}\delta_{Y_\varepsilon}(x)$, as there are $\frac{2}{3}\cdot\size{C}$ terminals in $\P'_C$ and the probability of each vertex of $C$ being one of the terminals is $\frac{2}{3}$. Moreover, since there are $\frac{2}{3}\cdot\size{C}$ edges in $\P'_C$, we also have $\EE{w(\P'_C)}=\frac{2}{3}w(C)$. Hence, we have 
\[
\EE{\delta_{Y_\varepsilon}(\P'_C)+n_\varepsilon w(\P'_C)} = \frac{2}{3}\sum_{x\in C}\delta_{Y_\varepsilon}(x)+\frac{2}{3}n_\varepsilon w(C).
\]

Since the algorithm selects the initial vertex to find $\P_C$ such that $\delta_{Y_\varepsilon}(\P_C)+n_\varepsilon w(\P_C)$ is minimized, it follows that 
\[
{\delta_{Y_\varepsilon}(\P_C)+n_\varepsilon w(\P_C)}\leq\EE{\delta_{Y_\varepsilon}(\P'_C)+n_\varepsilon w(\P'_C)}= \frac{2}{3}\sum_{x\in C}\delta_{Y_\varepsilon}(x)+\frac{2}{3}n_\varepsilon w(C).
\]
Recall that $\P_{X_\varepsilon}=\bigcup_{C\in\C} \P_C$. Thus, we have
\begin{align*}
\delta_{Y_\varepsilon}(\P_{X_\varepsilon})+n_\varepsilon w(\P_{X_\varepsilon})&\leq\sum_{C\in\C}\lrA{\frac{2}{3}\sum_{x\in C}\delta_{Y_\varepsilon}(x)+\frac{2}{3}n_\varepsilon w(C)}\\
&=\frac{2}{3}\delta_{Y_\varepsilon}(X_\varepsilon)+\frac{2}{3}n_\varepsilon w(\C)\\
&\leq\frac{2}{3}\delta_{Y_\varepsilon}(X_\varepsilon)+\frac{2}{3}\lrA{\frac{1}{2}\delta_{Y_\varepsilon}(X_\varepsilon)+\frac{3}{4}n_\varepsilon w(\C_{X_\varepsilon})}\\
&=\delta_{Y_\varepsilon}(X_\varepsilon)+\frac{1}{2}n_\varepsilon w(\C_{X_\varepsilon}),
\end{align*}
where the first equality follows from the fact that $\C$ is a cycle packing in $G[X_\varepsilon]$ and the second inequality follows from Lemma~\ref{f1}.

Similarly, there also exists a $O(n^3)$-time algorithm to get a 3-path packing $\P_{Y_\varepsilon}$ in $G[Y_\varepsilon]$ with $\delta_{X_\varepsilon}(\P_{Y_\varepsilon})+n_\varepsilon w(\P_{Y_\varepsilon})\leq \delta_{X_\varepsilon}(Y_\varepsilon)+\frac{1}{2}n_\varepsilon w(\C_{Y_\varepsilon})$.
\end{proof}

\begin{theorem}\label{main}
For BTTP with any $n$ and any constant $\varepsilon>0$, there is a randomized $(3/2+\varepsilon)$-approximation algorithm with a running time of $O(n^3)$.
\end{theorem}
\begin{proof}
By Theorem~\ref{3-path}, we can get 3-path packings $\P_{X_\varepsilon}$ in $G[X_\varepsilon]$ and $\P_{Y_\varepsilon}$ in $G[Y_\varepsilon]$ in $O(n^3)$ time with $\delta_{Y_\varepsilon}(\P_{X_\varepsilon})+n_\varepsilon w(\P_{X_\varepsilon})+\delta_{X_\varepsilon}(\P_{Y_\varepsilon})+n_\varepsilon w(\P_{Y_\varepsilon})\leq \delta_{Y_\varepsilon}(X_\varepsilon)+\frac{1}{2}n_\varepsilon w(\C_{X_\varepsilon})+\delta_{X_\varepsilon}(Y_\varepsilon)+\frac{1}{2}n_\varepsilon w(\C_{Y_\varepsilon})$. 

If $n\geq 108d^3$, by Lemma~\ref{lb2} and Theorem~\ref{res-construction}, we can obtain a solution in $O(n^2)$ time with an expected weight of at most $(3/2+\varepsilon)\cdot\OPT$ by applying $\P_{X_\varepsilon}$ and $\P_{Y_\varepsilon}$ to our 3-path construction and using the randomized labeling algorithm. Otherwise, as mentioned, we have $n< 108d^3=O_\varepsilon(1)$, and in this case, we can solve the problem optimally in constant time. Therefore, there is a randomized $(3/2+\varepsilon)$-approximation algorithm for BTTP, and the total running time is $O(n^3)$.
\end{proof}

Our algorithm is randomized because it uses a simple randomized labeling technique.
It can be derandomized in polynomial time using the well-known method of conditional expectations in~\cite{williamson2011design}, while preserving the approximation ratio. 
To show this, by the proof of Theorem~\ref{res-construction}, we only need to show how to derandomize the algorithm for $n\geq 108d^3$, i.e., how to find in polynomial time a labeling for 3-paths in $\P_{X_\varepsilon}$ and $\P_{Y_\varepsilon}$ such that using this labeling the 3-path construction generates a solution with a weight of at most 
\[
UB\coloneqq\delta_{Y_\varepsilon}(\P_{X_\varepsilon})+n_\varepsilon w(\P_{X_\varepsilon})+\delta_{X_\varepsilon}(\P_{Y_\varepsilon})+n_\varepsilon w(\P_{Y_\varepsilon})+\varepsilon\cdot\OPT.
\]

In the following, we demonstrate the main idea of the derandomization using the method of conditional expectations. Some detailed examples can be found in~\cite{miyashiro2012approximation,zmor}.

Recall that there are $md$ 3-paths in $\P_{X_\varepsilon}$ (resp., $\P_{Y_\varepsilon}$), which will be labeled using $\{P_1,..., P_{md}\}$ (resp., $\{Q_1,...,Q_{md}\}$). Previously, the algorithm uses a randomized labeling method such that for each $i\in\{1,...,$ $md\}$ each 3-path in $\P_{X_\varepsilon}$ (resp., $\P_{Y_\varepsilon}$ ) is labeled as $P_i$ (resp., $Q_i$) with the same probability of $\frac{1}{md}$. Now, we determine the \emph{random decision} on $P_1,..., P_{md}$ sequentially. 

Suppose that we have determined the labels $P_1,...,P_s$, where $0\leq s<md$, such that by randomly labeling the 3-paths in $\P_{X_\varepsilon}\setminus\{P_1,...,P_s\}$ (resp., $\P_{Y_\varepsilon}$) using $\{P_{s+1},..., P_{md}\}$ (resp., $\{Q_1,...,Q_{md}\}$) the 3-path construction generates a solution with an expected weight of $\EE{Sol\mid P_1,...,P_s}\leq UB$. Now, we want to determine the label of $P_{s+1}$ such that by randomly labeling the 3-paths in $\P_{X_\varepsilon}\setminus\{P_1,...,P_{s+1}\}$ (resp., $\P_{Y_\varepsilon}$) using $\{P_{s+2},$ $..., P_{md}\}$ (resp., $\{Q_1,...,Q_{md}\}$) the 3-path construction generates a solution with an expected weight of $\EE{Sol\mid P_1,...,P_{s+1}}\leq UB$. It can be done simply by letting 
\[
P_{s+1}=\text{arg min}_{P_{s+1}\in \P_{X_\varepsilon}\setminus\{P_1,...,P_s\}}\EE{Sol\mid P_1,...,P_{s+1}}.
\]

When calculating the conditional expected weight, we can still use the assumptions made in the proof of Theorem~\ref{res-construction}, since $UB$ is derived based on these assumptions. Thus, it is easy to check that the conditional expected weight can be computed in polynomial time, and then by recursion, the labels of 3-paths in $\P_{X_\varepsilon}$ can be determined in polynomial time.
Similarly, by sequentially determining the random decision on $Q_1,...,$ $Q_{md}$, the labels of 3-paths in $\P_{Y_\varepsilon}$ can also be determined in polynomial time.

\section{Application}\label{SC.5}
To test the performance of our algorithm, we introduce a new BTTP instance and apply our algorithm to it. 
Since the 3-path construction in Theorem~\ref{feasibility} requires $n$ to be sufficiently large, we instead use the 3-path construction in Theorem~\ref{feasibility+} for the experiments.
Additionally, we extend our 3-path construction to a \emph{3-cycle construction} by slightly modifying the design of normal super-games, which turns out to be more practical when $n$ is small.

\textbf{The New Instance.}
It is motivated by the real situation of NBA.
Since 2004, the number of teams in NBA has always been 30, where there are 15 teams in the Western Conference and 15 teams in the Eastern Conference. 
Thus, Hoshino and Kawarabayashi~\cite{Bipartite-conference} constructed an NBA instance with $n=15$ for BTTP.
In recent years, there have been rumors saying that NBA is poised to expand to 32 teams, with the potential inclusion of two new teams from Las Vegas and Seattle. Therefore, we create an instance, where we assume that two new teams from Las Vegas and Seattle join in the Western Conference, and also the Minnesota Timberwolves, originally from the Western Conference, are moved to join the Eastern Conference.
There are 32 teams in total, and hence this is an instance for BTTP with $n=16$. 
Note that the strict conditions of BTTP are not part of the NBA scheduling requirement, as evidenced by the San Antonio Spurs playing 6 consecutive home games followed by 8 consecutive away games during the 2009-10 regular season~\cite{Bipartite-conference}, as well as the NBA scheduling games on nonconsecutive days.
Next, we list their names, respectively.

The 16 teams in the Western Conference are Dallas Mavericks, Denver Nuggets, Golden State Warriors, Houston Rockets, LA Clippers, Los Angeles Lakers, Memphis Grizzlies, New Orleans Pelicans, Oklahoma City Thunder, Phoenix Suns, Portland Trail Blazers, Sacramento Kings, San Antonio Spurs, Utah Jazz, Las Vegas Team, and Seattle Team.

The 16 teams in the Eastern Conference are Atlanta Hawks, Boston Celtics, Brooklyn Nets, Charlotte Hornets, Chicago Bulls, Cleveland Cavaliers, Detroit Pistons, Indiana Pacers, Miami Heat, Milwaukee Bucks, Minnesota Timberwolves, New York Knicks, Orlando Magic, Philadelphia 76ers, Toronto Raptors, and Washington Wizards.

The home venues for the existing 30 NBA teams are available on a website\footnote{\url{https://en.wikipedia.org/wiki/List_of_National_Basketball_Association_arenas}}. The home venues for the Las Vegas and Seattle teams are estimated by considering the current arenas within their respective states. The detailed information of our new instance, including names, home venues, and the coordinates of the home venues, can be found in Tables~\ref{west} and \ref{east}, where we number teams in the Western Conference as 1,2,...,16, and teams in the Eastern Conference as 17,18,...,32, respectively. An illustration of the locations of the 32 teams in our instance can be seen in Figure~\ref{fig04}\footnote{It is adapted from the figure available at \url{https://gist.github.com/bordaigorl/fce575813ff943f47505}.}.

\begin{table}[ht]
\centering
\begin{tabular}{c|c|c|c}
\hline
& Team Names & Home Venues & Coordinates (Latitude, Longitude) \\
\hline
\circled{1} & \text{Dallas Mavericks} & \text{American Airlines Center} & (32.7904, -96.8103) \\
\circled{2} & \text{Denver Nuggets} & \text{Ball Arena} & (39.7487, -105.0076) \\
\circled{3} & \text{Golden State Warriors} & \text{Chase Center} & (37.7680, -122.3879) \\
\circled{4} & \text{Houston Rockets} & \text{Toyota Center} & (29.7509, -95.3622) \\
\circled{5} & \text{LA Clippers} & \text{Crypto.com Arena} & (34.0430, -118.2673) \\
\circled{6} & \text{Los Angeles Lakers} & \text{Crypto.com Arena} & (34.0430, -118.2673) \\
\circled{7} & \text{Memphis Grizzlies} & \text{FedExForum} & (35.1382, -90.0506) \\
\circled{8} & \text{New Orleans Pelicans} & \text{Smoothie King Center} & (29.9490, -90.0821) \\
\circled{9} & \text{Oklahoma City Thunder} & \text{Paycom Center} & (35.4634, -97.5151) \\
\circled{10} & \text{Phoenix Suns} & \text{Footprint Center} & (33.4457, -112.0712) \\
\circled{11} & \text{Portland Trail Blazers} & \text{Moda Center} & (45.5316, -122.6668) \\
\circled{12} & \text{Sacramento Kings} & \text{Golden 1 Center} & (38.5802, -121.4997) \\
\circled{13} & \text{San Antonio Spurs} & \text{Frost Bank Center} & (29.4270, -98.4375) \\
\circled{14} & \text{Utah Jazz} & \text{Delta Center} & (40.7683, -111.9011) \\
\circled{15} & \text{Las Vegas Team} & \text{T-Mobile Arena} & (36.1028, -115.1782) \\
\circled{16} & \text{Seattle Team} & \text{Climate Pledge Arena} & (47.6221, -122.3541) \\
\end{tabular}
\caption
{
Information of 16 Teams in the Western Conference - Team Names, Home Venues, and Home Venue Coordinates.
}
\label{west}
\end{table}

\begin{table}[ht]
\centering
\begin{tabular}{c|c|c|c}
\hline
& Team Names & Home Venues & Coordinates (Latitude, Longitude) \\
\hline
\circled{17} & \text{Atlanta Hawks} & \text{State Farm Arena} & (33.7573, -84.3963) \\
\circled{18} & \text{Boston Celtics} & \text{TD Garden} & (42.3662, -71.0621) \\
\circled{19} & \text{Brooklyn Nets} & \text{Barclays Center} & (40.6826, -73.9754) \\
\circled{20} & \text{Charlotte Hornets} & \text{Spectrum Center} & (35.2252, -80.8393) \\
\circled{21} & \text{Chicago Bulls} & \text{United Center} & (41.8807, -87.6742) \\
\circled{22} & \text{Cleveland Cavaliers} & \text{Rocket Mortgage FieldHouse} & (41.4965, -81.6880) \\
\circled{23} & \text{Detroit Pistons} & \text{Little Caesars Arena} & (42.3411, -83.0552) \\
\circled{24} & \text{Indiana Pacers} & \text{Gainbridge Fieldhouse} & (39.7640, -86.1555) \\
\circled{25} & \text{Miami Heat} & \text{Kaseya Center} & (25.7814, -80.1870) \\
\circled{26} & \text{Milwaukee Bucks} & \text{Fiserv Forum} & (43.0451, -87.9173) \\
\circled{27} & \text{Minnesota Timberwolves} & \text{Target Center} & (44.9795, -93.2760) \\
\circled{28} & \text{New York Knicks} & \text{Madison Square Garden} & (40.7505, -73.9934) \\
\circled{29} & \text{Orlando Magic} & \text{Kia Center} & (28.5381, -81.3841) \\
\circled{30} & \text{Philadelphia 76ers} & \text{Wells Fargo Center} & (39.9012, -75.1720) \\
\circled{31} & \text{Toronto Raptors} & \text{Scotiabank Arena} & (43.6435, -79.3790) \\
\circled{32} & \text{Washington Wizards} & \text{Capital One Arena} & (38.8982, -77.0208) \\
\end{tabular}
\caption
{
Information of 16 Teams in the Eastern Conference - Team Names, Home Venues, and Home Venue Coordinates.
}
\label{east}
\end{table}

\begin{figure}[ht]
    \centering
    \includegraphics[scale=0.6]{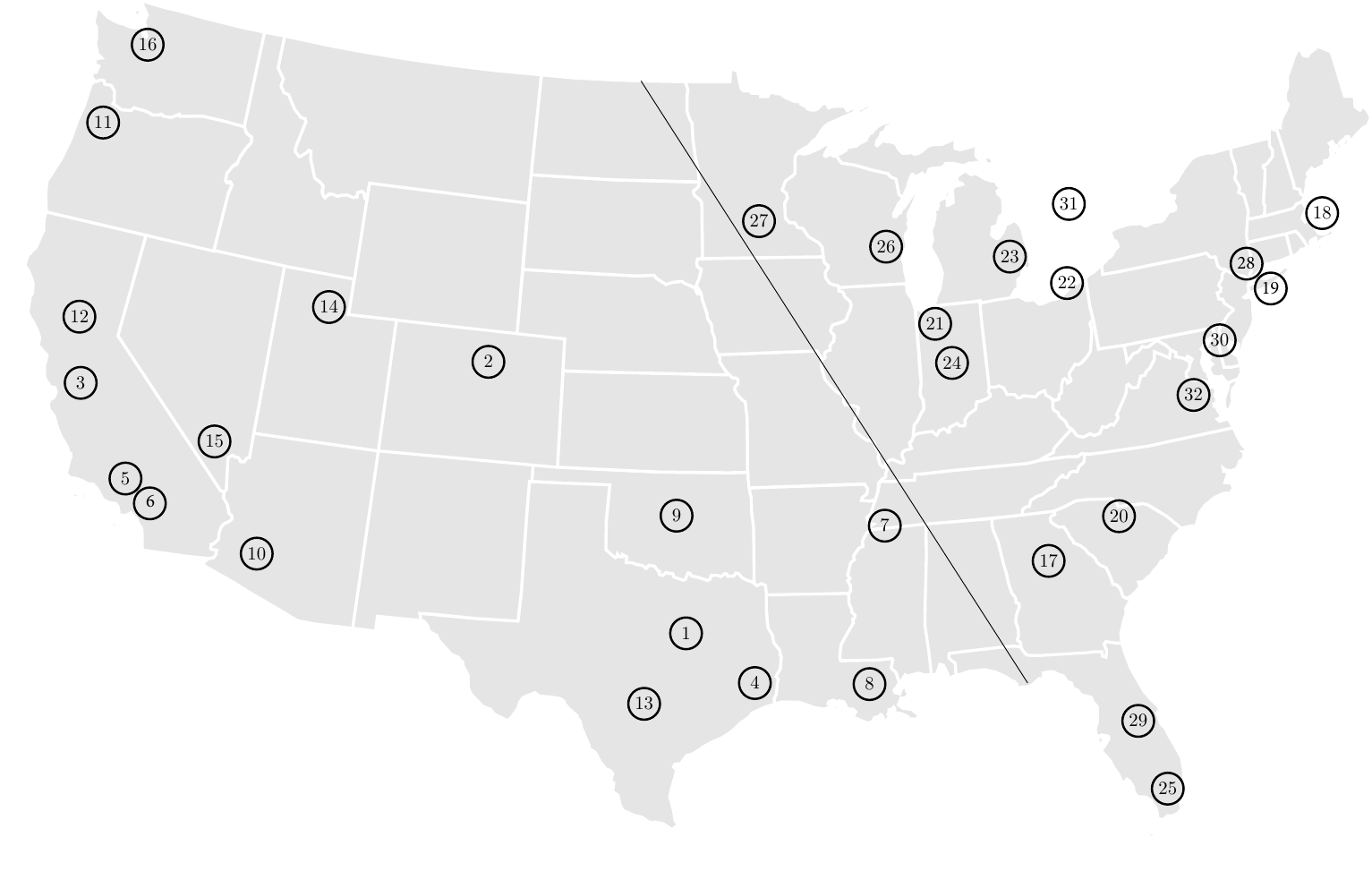}
    \caption{An illustration of the locations of the 32 teams in our instance.}
    \label{fig04}
\end{figure}

We determined their distance matrix using the following method: first, we surveyed the home venues of the existing 30 teams and made educated guesses about the home venues of the two new teams; then, we obtained the longitude and latitude coordinates of these home venues using Google Maps; last, we calculated the distances (in miles, using double precision) between any pair of teams based on the coordinates using the Haversine formula~\cite{sinnott1984virtues}. 

\textbf{The Parameters.}
We set $d=1$ and $m=5$. Then, we have $l=1$, and the games in the last time slot may violate the no-repeat constraint. As in the proof of Theorem~\ref{feasibility+}, if we denote the games on $2n$ days in our 3-path construction as $g_1g_2\dots g_{2n}$, we can rearrange them as $g_{2n}g_1g_2\dots g_{2n-1}$ to fix this problem.

\textbf{The 3-Cycle Construction.} Our 3-path construction can be easily extended to pack 3-cycles. The main difference is that the design of normal super-games becomes different.

For the case of $d=1$, we modify the design of normal super-games as follows. Assume that there is a normal super-game between $S=\{s_0,s_1,s_2\}$ and $T=\{t_0,t_1,t_2\}$.
In the extension, we define $m_i\coloneqq\{s_{i'}\rightarrow t_{(i'+i)\bmod 3}\}_{i'=0}^{2}$, and arrange games on 6 days in the order of $m_0m_1m_2\overline{m_0m_1m_2}$. An illustration of the regular games after extending one normal super-game for the 3-cycle construction is shown in Table~\ref{normal-cycle}.

\begin{table}[ht]
\centering
\begin{tabular}{c|*{6}{c}}
  & $0$ & $1$ & $2$ & $3$ & $4$ & $5$\\
\hline
  $s_0$ & $t_0$ & $t_1$ & $t_2$ & $\pmb{t_0}$ & $\pmb{t_1}$ & $\pmb{t_2}$\\
  $s_1$ & $t_1$ & $t_2$ & $t_0$ & $\pmb{t_1}$ & $\pmb{t_2}$ & $\pmb{t_0}$\\
  $s_2$ & $t_2$ & $t_0$ & $t_1$ & $\pmb{t_2}$ & $\pmb{t_0}$ & $\pmb{t_1}$\\
\hline
  $t_0$ & $\pmb{s_0}$ & $\pmb{s_2}$ & $\pmb{s_1}$ & $s_0$ & $s_2$ & $s_1$\\
  $t_1$ & $\pmb{s_1}$ & $\pmb{s_0}$ & $\pmb{s_2}$ & $s_1$ & $s_0$ & $s_2$\\
  $t_2$ & $\pmb{s_2}$ & $\pmb{s_1}$ & $\pmb{s_0}$ & $s_2$ & $s_1$ & $s_0$\\
\end{tabular}
\caption{
Extending the normal super-game for the 3-cycle construction between $\{s_0,s_1,s_2\}$ and $\{t_0,t_1,t_2\}$ into regular games on $6$ days, where home games are marked in bold.
}
\label{normal-cycle}
\end{table}

In the extension of the new normal super-game, each team in $S$ (resp., $T$) plays 3-consecutive away games along two edges in the 3-cycle $t_0t_1t_2t_0$ (resp., $s_0s_1s_2s_0$). 
Moreover, on each day teams in the same league either all play home games or all play away games. Due to this, before extending one new normal super-game between super-teams $S$ and $T$, we may relabel teams in $S$ and teams $T$. There are $3!\times 3!=36$ ways, and we choose the best one so that the total traveling distance of teams in $S\cup T$ is minimized on the extended 6 days. We will recover their labels after extending the new normal super-game, and the feasibility still holds. 
Note that this property does not hold for the previous normal super-game.

For the case of $d>1$, the normal super-game between $S=\{s_0,s_1,s_2,...,s_{3d-3},s_{3d-2},s_{3d-1}\}$ and $T=\{t_0,t_1,t_2,...,t_{3d-3},t_{3d-2},t_{3d-1}\}$ can be reduced to the case of $d=1$ as follows. 
For each $i\in\{0,1,...,d-1\}$, let $SC_i=\{s_{3i},s_{3i+1},s_{3i+2}\}$ and $TC_i=\{t_{3i},t_{3i+1},t_{3i+2}\}$, which denote 3-cycles with respect to $S$ and $T$, respectively. Therefore, we have $S=\bigcup_{0\leq i<d}SC_i$ and $T=\bigcup_{0\leq j<d}TC_j$. In the extension, there are $d$ \emph{sub-slots}. In the $j$th \emph{sub-slot} where $j\in\{0,1,...,d-1\}$, we arrange the pairs $\{SC_i-TC_{(i+j)\bmod d}\}_{i=0}^{d-1}$, where $SC_i-TC_{(i+j)\bmod d}$ corresponds to the normal super-game in the case of $d=1$ and will be further extended to regular games on 6 days. Thus, the normal super-game will be extend to regular games on $6d$ days.

\begin{remark}
The 3-cycle construction mainly modifies the design of normal super-games, and therefore, it requires the same assumptions on $m$ and $d$. It may be more practical for small $n$ because the design of its normal super-games also reduces the frequency of returning homes for involved regular teams.
\end{remark}

\textbf{Optimizing Left Super-games.} The design of left super-games in our 3-path and 3-cycle constructions can be optimized. Recall that we arrange the games in left super-game as $m_1\overline{m_2}m_3\overline{m_1}m_2\overline{m_3}$. To reduce the frequency of returning homes, we rearrange them as $m_1m_2\overline{m_3m_1}m_3\overline{m_2}$ for the 3-path construction, and $m_1m_2m_3\overline{m_1m_2m_3}$ for the 3-cycle construction. The feasibility also holds.

\textbf{The Implementations.}
Instead of finding a 3-path packing using the algorithm in Theorem~\ref{3-path}, we directly label teams randomly, and try to improve our schedule by exchanging the labels of two teams in the same league. 
There are $\binom{n}{2}$ pairs of teams in $X$ and $\binom{n}{2}$ pairs of teams in $Y$.
We consider these $n(n-1)$ pairs in a random order.
From the first pair to the last, we test whether the weight of our schedule can be reduced after we swap the labels of teams in the pair. If no, we do not swap them and go to the next pair. If yes, we swap them and go to the next pair. After considering all the $n(n-1)$ pairs, if there is an improvement, we repeat the whole procedure. Otherwise, the procedure ends.

The main reasons for using the local search method are as follows.
First, local search is easy to implement because it does not require the computation of the minimum weight cycle packing, yet it generally produces schedules of very high quality.
Second, even with the local search method, the algorithms run quickly since the sizes of TTP/BTTP instances are typically small (e.g., TTP benchmark instances~\cite{trick2007challenge,DBLP:journals/eor/BulckGSG20} have a maximum size of only 40).
Finally, the local search method primarily involves exchanging the labels of two teams without altering the overall structure of the 3-path/cycle construction. As a result, the resulting schedules are ``isomorphic'', preserving the same underlying structure as the original 3-path/cycle construction, with the final schedule reflecting the novelty of the construction.

\begin{remark}
While local search theoretically sacrifices the approximation guarantee, we believe that in practice, it can yield significantly better solutions.
\end{remark}

Our algorithms are coded in C++, on a standard desktop computer with a 3.20GHz AMD Athlon 200GE CPU and 8 GB RAM. 
The details of our algorithms can be found in \url{https://github.com/JingyangZhao/BTTP}.

To better highlight the quality of our experimental results, we compare them with the well-known \emph{Independent Lower Bound}~\cite{easton2003solving} , as it is of high quality and remains the best-known lower bound for several large benchmark instances of TTP~\cite{DBLP:journals/eor/BulckGSG20}.
The idea is to compute the best-possible traveling distance of a single team $v$ (i.e., the traveling distance of its minimum weight itinerary) independently without considering the feasibility of other teams. The value is denoted as $ILB_v$, and the Independent Lower Bound is then defined as $ILB\coloneqq\sum_{v\in X\cup Y}ILB_v$. 
It is easy to observe that the lower bounds in Lemmas~\ref{lb1} and \ref{lb2} are not stronger than the Independent Lower Bound, as these lower bounds are derived by analyzing the lower bound of each team's (optimal) itinerary and summing them, and hence they are further lower bounds of the Independent Lower Bound. 
Note that $ILB_v$ can be obtained by solving the Capacitated Vehicle Routing Problem (CVRP)~\cite{dantzig1959truck}. For example, if $v\in X$, we obtain a CVRP instance, where a vehicle with a capacity of 3 is located at $v$, each team in $Y$ has a unit demand, and the goal is to find a minimum itinerary for the vehicle to fulfill the demand of the teams in $Y$.
To solve this CVRP instance, we simply use a brute-force enumeration since the instance size is small. We have $ILB=655477.159$. Note that the brute-force enumeration takes approximately half an hour.

Our results can be seen in Table~\ref{experiment}, where the column `\emph{Construction}' indicates the 3-path and the 3-cycle constructions; `\emph{Result}' lists the results of our algorithms; `\emph{Gap}' is defined as $\frac{Result~-~ILB}{ILB}$, and `\emph{Time}' is the running time of our algorithms.

\begin{table}
\centering
\begin{tabular}{c|c|c|c}
\textbf{Construction} & \textbf{Result} & \textbf{Gap} & \textbf{Time}\\
\hline
3-path & 817088.498 & 24.66\% & 0.85s\\
\hline 
3-cycle & 717174.266 & 9.42\% & 1.03s\\
\end{tabular}
\caption{Experimental results of our 3-path and 3-cycle constructions.}
\label{experiment}
\end{table}

We can see that both of our algorithms run very fast. Compared to $ILB$, the quality of the 3-path construction has a gap of $24.66\%$, which is much smaller than the desired 1.5-approximation ratio. Moreover, the 3-cycle construction is more practical, which can reduce the gap to only $9.42\%$.

\section{Conclusion}\label{SC.6}
In this paper, for BTTP with any $n$ and any constant $\varepsilon>0$, we propose a $(3/2+\varepsilon)$-approximation algorithm, which significantly improves the previous result.
Our theoretical result relies on three key ideas: the first is a new 3-path construction, the second is a new lower bound, and the third is a 3-path packing algorithm. Our methods also have the potential to design better approximation algorithms for TTP, which is left for future work. 
For applications, we create a new real instance from NBA for BTTP. Experimental results show that the 3-cycle construction has a very good practical performance.
Our method is also possibly used to solve BTTP with the maximum number of consecutive home-games or away-games bounded by $k>3$. However, for larger $k$, it is not easy to compute a $k$-path packing good enough to obtain a desired approximation ratio and the analysis of similar lower bounds becomes more complicated.

\section*{Acknowledgments}
A preliminary version of this paper was presented at the 33rd International Joint Conference on Artificial Intelligence (IJCAI 2024)~\cite{DBLP:conf/ijcai/0001024a}.

\bibliographystyle{plain}
\bibliography{zmain}
\end{document}